\documentclass[journal]{new-aiaa}
\usepackage[utf8]{inputenc}
\usepackage{textcomp}
\usepackage{graphicx}
\usepackage{amsmath}
\usepackage[version=4]{mhchem}
\usepackage{siunitx}
\usepackage{longtable,tabularx}
\DeclareMathOperator{\sign}{sgn}
\DeclareMathOperator*{\argmax}{arg\,max}
\usepackage{accents}

\renewcommand{\bot}[1]{\accentset{\circ}{#1}}

\usepackage{amsthm}
\usepackage{physics}
\usepackage{booktabs}
\usepackage{multirow}
\usepackage{subcaption}
\usepackage{color}

\usepackage{bm}

\newtheorem{proposition}{Proposition}

\title{An Optimal Bearing-Only-Information Strategy for\\ Unmanned Aircraft Collision Avoidance}

\author{Timothy L.~Molloy \footnote{Advance Queensland Research Fellow, School of Electrical Engineering and Computer Science; t.molloy@qut.edu.au}}
\affil{Queensland University of Technology (QUT), Brisbane, Queensland, 4000, Australia}
\author{Tristan Perez \footnote{Research Program Lead, Autonomous Systems.}}
\author{Brendan P.~Williams \footnote{Associate Technical Fellow, Autonomous Systems.}}
\affil{Boeing Research \& Technology Australia, St Lucia, Queensland, 4072, Australia}

\newcommand\blfootnote[1]{%
  \begingroup
  \renewcommand\thefootnote{}\footnote{#1}%
  \addtocounter{footnote}{-1}%
  \endgroup
}

\begin{document}

\maketitle

\blfootnote{Copyright © 2020 by the American Institute of Aeronautics and Astronautics, Inc. All rights reserved.}

\begin{abstract}
This paper presents a novel collision avoidance strategy for unmanned aircraft detect and avoid that requires only information about the relative bearing angle between an aircraft and hazard.
It is shown that this bearing-only strategy can be conceived as the solution to a novel differential game formulation of collision avoidance, and has several intuitive properties including maximising the instantaneous range acceleration in situations where the hazard is stationary or has a finite turn rate.
The performance of the bearing-only strategy is illustrated in simulations based on test cases drawn from draft minimum operating performance standards for unmanned aircraft detect and avoid systems.
\end{abstract}

% \section*{Nomenclature}

% \noindent(Nomenclature entries should have the units identified)

% {\renewcommand\arraystretch{1.0}
% \noindent\begin{longtable*}{@{}l @{\quad=\quad} l@{}}
% $A$  & amplitude of oscillation \\
% $a$ &    cylinder diameter \\
% $C_p$& pressure coefficient \\
% $Cx$ & force coefficient in the \textit{x} direction \\
% $Cy$ & force coefficient in the \textit{y} direction \\
% c   & chord \\
% d$t$ & time step \\
% $Fx$ & $X$ component of the resultant pressure force acting on the vehicle \\
% $Fy$ & $Y$ component of the resultant pressure force acting on the vehicle \\
% $f, g$   & generic functions \\
% $h$  & height \\
% $i$  & time index during navigation \\
% $j$  & waypoint index \\
% $K$  & trailing-edge (TE) nondimensional angular deflection rate\\
% $\Theta$ & boundary-layer momentum thickness\\
% $\rho$ & density\\
% \multicolumn{2}{@{}l}{Subscripts}\\
% cg & center of gravity\\
% $G$ & generator body\\
% iso	& waypoint index
% \end{longtable*}}

\section{Introduction}

\lettrine{R}{outine} operations of unmanned aircraft in non-segregated airspace are currently restricted due to the ongoing challenge of ensuring that they will not increase or disrupt the safety risk of other users sharing the airspace \cite{Clothier2015}.
Detect and avoid (DAA) systems are a key technology for controlling the risk of mid-air collision posed by unmanned aircraft \cite{Yu2015,Clothier2015,Mcfadyen2016,Prats2012}.
The development of DAA systems capable of satisfying emerging minimum operating performance standards (e.g., \cite{RTCA2016}), however, remains a significant challenge, with a key open problem being the selection of collision avoidance manoeuvres given only partial information about the state (i.e. the position and velocity) of collision hazards \cite{Yu2015,Prats2012}.
Indeed, despite the potential for DAA systems to fuse information from multiple sensors to obtain complete hazard state information (cf.~\cite{Manfredi2016,FASANO2015,Ramasamy2014}), the time-critical nature of collision avoidance and the potential for failures or delays in sensor, communication, and navigation systems (including late or delayed hazard detections) implies that DAA systems will face situations in which they must select and commence avoidance manoeuvres with only partial hazard state information \cite{Yu2015,Prats2012,Tabassum2019}.
In this paper, we therefore aim to identify a novel collision avoidance strategy that uses only information about the relative bearing angles to collision hazards (i.e. relative azimuth and elevation angles).
We specifically seek to identify a bearing-only avoidance strategy with optimality properties under a differential game formulation of the collision avoidance problem so as to establish a novel description of the theoretical performance achievable given only bearing information.

% Setting the Scene for Collision Avoidance
% - Layers
% - ACAS-X and Xu
% - Different Sensor Technologies
\subsection{Motivation}
The International Civil Aviation Organization (ICAO) in \cite{ICAO2015} describes a three-stage approach to unmanned aircraft DAA based on the distance and/or time remaining before a collision.
The three stages (illustrated in Fig.~\ref{fig:layers}) mirror those for manned aircraft (cf.~\cite{ICAO2005a}) and are: strategic conflict management, separation provision (or remain well clear), and collision avoidance.
The strategic conflict management stage is divorced from the specific capabilities of DAA systems and included to emphasise the importance of safe flight planning (e.g., airspace management and traffic synchronisation).
In the separation provision stage, air traffic control and/or DAA systems are expected to manage the tactical process of keeping aircraft away from hazards by at least an appropriate minimum separation time and/or distance \cite{ICAO2015}.
In both the strategic conflict management and separation provision phases, it is therefore possible and appropriate to resolve conflicts with potential hazards whilst taking into account flight objectives such minimising fuel burn or deviation from a desired path.
In the collision avoidance stage however, manoeuvres of ``last resort'' are to be selected and performed (through the use of DAA systems) with the sole aim of preventing an impending collision \cite{ICAO2015}.

\begin{figure}[t!]
    \centering
    \includegraphics[width = 3.2in]{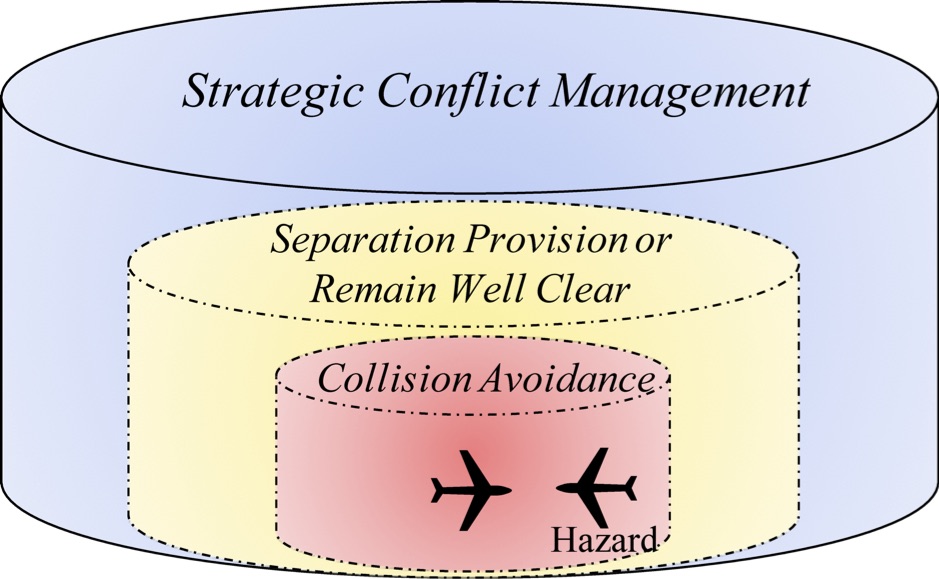}
    \caption{Staged Approach to Detect and Avoid (DAA) described in \cite{ICAO2015}.}
    \label{fig:layers}
\end{figure}

The specifics of when each of the three stages of DAA described in \cite{ICAO2015} is applicable are yet to be finalised, but have been defined in early draft minimum operating performance standards for DAA systems (e.g.~\cite{RTCA2016}) so as to be compatible with the next-generation airborne collision avoidance system (ACAS X) \cite{Kochenderfer2012}.
ACAS X and its unmanned aircraft variant ACAS Xu therefore represent promising candidate technologies for unmanned aircraft DAA \cite{Deaton2020,Manfredi2016}.
However, by the nature of its implementation, ACAS X is reliant on the equipage of the traffic pair with Automatic Dependent Surveillance Broadcast (ADS-B) sensors and so it is only effective for avoiding cooperative and semi-cooperative hazards.
Whilst ACAS Xu promises an ability to avoid noncooperative hazards by using a suite of sensors (e.g., ADS-B, electro-optical, infrared, acoustic, lidar and/or radar sensors), carrying more than one sensor is undesirable (if not impractical) for small to medium unmanned aircraft due to Size, Weight, and Power (SWaP) constraints.
Furthermore, ACAS Xu determines avoidance manoeuvres using optimisation tables that are precomputed numerically offline, and so its deployment raises significant challenges associated with how to \emph{a priori} select discretizations of the state and action (i.e., manoeuvre) spaces together with dynamic models of different hazards and sensor configurations so that the system has suitable available precomputed manoeuvres \cite{Julian2019,Manfredi2016}.

In light of the limitations of ACAS X and ACAS Xu as DAA systems, there remains a strong practical and theoretical impetus to investigate the fundamental performance attainable with a range of standalone sensor technologies for DAA --- especially for small to medium unmanned aircraft.
Electro-optical, infrared, and acoustic sensors all represent strong candidates for DAA due to their ability to detect noncooperative hazards, and since they typically have lower SWaP requirements compared to competing lidar and radar sensors \cite{Yu2015,Mcfadyen2016,Opromolla2019}.
Whilst electro-optical, infrared, and acoustic sensors are capable of producing range estimates if hazards are tracked for prolonged periods of time (cf.~\cite{Molloy2014a,Choi2013,Dippold2009}) or if multiple sensors are employed simultaneously (e.g. stereo vision \cite{Ramani2017}), they only provide direct instantaneous measurements of the bearing angles to potential collision hazards (i.e., azimuth and elevation angles, cf.~\cite{Yu2015,Molloy2017a,James2018a,James2018}).
The time-critical nature of collision avoidance and the potential for failures in sensors therefore means that DAA systems comprised in any way of electro-optical, infrared, and acoustic sensors will almost certainly face situations in which they are forced to select and commence avoidance manoeuvres using only bearing information.
In this paper, we therefore seek to investigate bearing-only collision avoidance for unmanned aircraft DAA.

\subsection{Related Work}
% Collision Avoidance
Outside of the specific context of unmanned aircraft DAA, the broader problem of collision avoidance has been studied extensively for all manner of agents including robots \cite{Dippold2009,Van2008,Gunasinghe2019}, aircraft \cite{Sharma2012,Smith2015,Jenie2015,Jenie2016,Tarnopolskaya2009,Tarnopolskaya2010,Exarchos2016,Choi2013,Cichella2018,Garden2017,FULTON2015,Fulton2018,Kuchar2000,Gunasinghe2019}, and marine surface craft \cite{Miloh1976,Merz1972,Merz1973,Johansen2016}.
The dominant approaches to collision avoidance have primarily involved the use of artificial potential fields (see \cite{Kuchar2000} and references therein), geometric arguments in either position or velocity space (e.g., the velocity obstacle methods of \cite{Van2008,Jenie2016,Jenie2015} and the Apollonius circle work of \cite{FULTON2015,Garden2017,Fulton2018}), or the solution of optimal control or differential game problems (see \cite{Tarnopolskaya2009,Tarnopolskaya2010,Exarchos2016,Miloh1976,Merz1972,Merz1973,Johansen2016,Mylvaganam2017,Jha2019,Julian2019} and references therein).
Whilst many of these existing approaches assume the availability of full hazard state information, considerable effort has recently been directed towards developing collision avoidance approaches that operate on the basis of partial hazard state information \cite{Choi2013,Cichella2018,Mcfadyen2016,Sharma2012,Dippold2009,Exarchos2014,Exarchos2015,Exarchos2016}.

Practical collision avoidance approaches operating on the basis of the partial state information provided by a single monocular electro-optical or infrared camera have been recently been developed across robotics and aerospace (see \cite{Choi2013,Cichella2018,Sharma2012,Dippold2009} and the survey paper \cite{Mcfadyen2016} with references therein).
The approaches detailed in \cite{Choi2013,Dippold2009,Mcfadyen2016} involve the use of specialised image processing techniques to extract image-based features (such as optical flow, hazard area expansion, relative bearing, or relative bearing rate) so that the range to the hazard can be estimated.
In contrast, the approaches detailed in \cite{Sharma2012,Cichella2018} use only measurements of the relative bearing to the hazard, but either assume bounds on the unknown range or introduce restrictive assumptions about the speed, shape, or heading of hazards.
The existing approaches of \cite{Choi2013,Cichella2018,Sharma2012,Dippold2009} therefore all fail to constitute pure bearing-only collision avoidance solutions and so suffer from the main drawbacks we have already identified including that delays in estimating range can render them unsuitable for (last resort) collision avoidance.
Furthermore, the approaches based on estimating range from bearing or bearing rate measurements are subject to fundamental limitations related to observability (cf.~\cite{Hepner1990} and \cite{He2018}) and those that rely on specialised image processing techniques can only be validated experimentally (making them less appealing than other approaches as a basis for DAA system development due to potential difficulties that can arise in the certification process).
These existing collision avoidance approaches developed for partial hazard state information therefore lack the rigorous theoretical performance descriptions and guarantees available for more mature collision avoidance approaches (that use full hazard state information) such as the velocity obstacles methods of \cite{Jenie2016,Jenie2015}, the Apollonius circle work of \cite{FULTON2015,Garden2017,Fulton2018}, or the optimal control and differential game approaches of \cite{Tarnopolskaya2009,Tarnopolskaya2010,Exarchos2016,Miloh1976,Merz1972,Merz1973,Mylvaganam2017}.
%The lack of rigorous theoretical performance descriptions and guarantees for these existing approaches therefore makes them less appealing than other approaches as a basis for DAA system development due to potential difficulties that can arise in the certification process.

Most recently, a collision avoidance approach that requires only knowledge of the hazard's speed and relative bearing has been proposed in \cite{Exarchos2016} on the basis of the solution to a novel differential game of pursuit and evasion (see also \cite{Exarchos2014,Exarchos2015}).
The approach of \cite{Exarchos2016} specifically seeks to provide manoeuvres such that a prespecified ``capture'' range threshold is maintained from hazards that are assumed to be agile in the sense of being able to instantaneously change their heading with the aim of reducing the range below the ``capture'' range threshold.
Whilst this pursuit-evasion formulation of collision avoidance is not uncommon in the literature, it differs from the foundational optimal-collision-avoidance formulation proposed by Merz in \cite{Merz1972,Merz1973} and extended and investigated by multiple authors since \cite{Miloh1976,Maurer2012,Tarnopolskaya2009,Tarnopolskaya2010}.
Specifically, the collision avoidance formulation of \cite{Merz1972,Merz1973} involves the maximisation of the minimum range (or miss-distance) from the hazard, which is a natural objective when collision is perceived to be immanent and avoidance action is required after all prespecified separation thresholds have been breached.
Motivated by the limited information required by the approach of \cite{Exarchos2016} and the natural miss-distance performance criterion of \cite{Merz1972,Merz1973}, in this paper, we pose and solve a new differential game formulation of collision avoidance by combining the miss-distance criterion of \cite{Merz1972,Merz1973} with the assumption of an agile hazard as in \cite{Exarchos2016}.
By doing so, we will identify a novel bearing-only collision avoidance strategy with theoretical performance descriptions and guarantees for unmanned aircraft DAA.

\subsection{Contributions}
The key contribution of this paper is the identification of a novel bearing-only collision avoidance strategy with optimality properties under a new differential-game formulation of collision avoidance.
The significance of this contribution is twofold:
\begin{enumerate}
    \item 
    The novel bearing-only strategy represents a crucial guidance technology for collision avoidance in future unmanned aircraft DAA systems that (either by fault or design) may only have access to a single sensor capable of directly measuring bearing angles (e.g. an electro-optical, infrared, or acoustic sensor); and,
    \item
    The development of the bearing-only strategy's optimality properties under a novel differential-game formulation of collision avoidance provides important new theoretical insight into the fundamental performance achievable using only bearing information.
\end{enumerate}
We specifically identify our novel bearing-only collision avoidance strategy by posing and solving a new differential-game formulation of collision avoidance with a miss-distance criterion similar to that in the optimal collision avoidance formulations of \cite{Merz1972,Merz1973,Miloh1976,Maurer2012,Tarnopolskaya2009,Tarnopolskaya2010}, and with an agile hazard similar to that of \cite{Exarchos2014,Exarchos2015,Exarchos2016}.
Complementing the existing formulations of optimal collision avoidance, we pose our differential game in three dimensions and solve it analytically.
Issues associated with the numeric solution of optimal control or differential game problems (such as those encountered by ACAS X and ACAS Xu, cf.~\cite{Julian2019}) are therefore sidestepped.
In this paper, we also provide an illustration of the practical performance of the bearing-only strategy through software-in-the-loop simulations at ranges within which state-of-the-art aircraft detection systems based on electro-optical, infrared, and acoustic sensors are capable of reliable aircraft detections.

The rest of this paper is organised as follows.
In Sec.~\ref{sec:problem}, we formulate collision avoidance in three-dimensions with an agile hazard as a differential game.
In Sec.~\ref{sec:diffGame}, we solve the differential game formulation of collision avoidance.
In Sec.~\ref{sec:properties}, we examine the properties of the solution to the differential game formulation of collision avoidance and recast it as a novel bearing-only collision avoidance strategy with additional optimality and robustness properties for non-agile hazards.
Finally, simulation results for the bearing-only collision avoidance strategy are presented in Sec.~\ref{sec:results} before conclusions are drawn in Sec.~\ref{sec:conclusion}.

\section{Problem Formulation}
\label{sec:problem}

Consider an unmanned aircraft and a ``hazard'' moving in three dimensions.
The (unmanned) aircraft and hazard are initially on closing trajectories with the range between them decreasing and a mid-air collision (or near mid-air collision) being immanent if neither manoeuvres.
The aircraft and hazard are assumed to maintain constant speeds $v_a$ and $v_h$, respectively.
The aircraft is however capable of changing direction by accelerating in a direction normal to its velocity, whilst the hazard is agile in the sense that it can instantaneously change the direction of its velocity vector.
The aircraft's objective is to avoid colliding with the hazard by maximising the minimum range (i.e., miss-distance).
In contrast, the hazard (either deliberately or accidentally) is assumed to manoeuvre so as to minimise the miss-distance.
%We aim to identify an optimal collision avoidance strategy for the aircraft that maximises the miss-distance in the presence of this worst case collision hazard.
We seek to identify a collision avoidance strategy for the aircraft that uses only information about the bearing angle to the hazard from the aircraft.

Our motivation for seeking a bearing-only collision avoidance strategy is to provide unmanned aircraft DAA systems with a capability to select manoeuvres of last resort to avoid impending collisions during the collision avoidance stage of DAA (as described in \cite{ICAO2015} and illustrated in Fig.~\ref{fig:layers}).
We shall therefore omit consideration of issues of secondary importance during collision avoidance and those that are only likely to have measurable impacts on longer time scales such as minimising flight path deviation or fuel burn, human factors or payload limits (e.g. if the aircraft has passengers), sensor field of view limitations (e.g. keeping the hazard in the sensor field of view), and the detailed performance and dynamics of the aircraft (e.g. if the aircraft's climb and dive rates are different).
Omission of these secondary and longer term effects is consistent with the previous optimal control and differential game treatments of collision avoidance (cf.~\cite{Tarnopolskaya2009,Tarnopolskaya2010,Exarchos2016,Maurer2012,Miloh1976,Merz1972,Merz1973,Mylvaganam2017}) and has been argued to lead to reasonable abstractions for guidance during close proximity collision encounters (see for example \cite{Maurer2012} and references therein).
Finally, consideration of many of these secondary issues would involve introducing extra constraints into our problem formulation, and so by omitting them, this paper will establish a new theoretical description of the optimal (unconstrained) performance achievable by bearing-only collision avoidance approaches.

\subsection{Three-Dimensional Equations of Motion}
We assume simplified point-mass dynamics for both the aircraft and hazard in three dimensions with the motion described in a coordinate system defined with reference to the aircraft as illustrated in Fig.~\ref{fig:fig3d}.
The relative position of the hazard at time $t \geq 0$ is described by the line-of-sight vector $\bm{R}(t) \in \mathbb{R}^3$ whilst the range between the hazard and the aircraft is given by the magnitude $r(t) \triangleq ||\bm{R}(t)||$.
The velocities of the aircraft and hazard are denoted by $\bm{V}_a(t) \in \mathbb{R}^3$ and $\bm{V}_h(t) \in \mathbb{R}^3$, respectively.
The bearing angle of the hazard as observed from the aircraft will be denoted by $\theta(t)$. Then, 
\begin{align}
    \label{eq:bearingToHazard}
    \cos \theta(t)
    &= \left< \bm{e}_{\bm{V}_a(t)}, \bm{e}_{\bm{R}(t)} \right>
\end{align}
where $\bm{e}_{\bm{x}}$ denotes a unit vector in the direction of a vector $\bm{x} \in \mathbb{R}^3$, and $\left< \cdot, \cdot \right>$ denotes the Euclidean inner product operator.
The bearing angle of the the hazard relative to the line-of-sight vector, denoted by $\phi(t)$, similarly satisfies
\begin{align*}
    \cos \phi(t)
    &= \left< \bm{e}_{\bm{V}_h(t)}, \bm{e}_{\bm{R}(t)} \right>
\end{align*}
and the angle between the aircraft and hazard velocities, denoted by $\psi(t)$, satisfies
\begin{align*}
    \cos \psi(t)
    &= \left< \bm{e}_{\bm{V}_a(t)}, \bm{e}_{\bm{V}_h(t)} \right>.
\end{align*}
The four variables $(r,\theta,\phi,\psi)$ are sufficient to describe the configuration of the system since we assume rotational symmetry around the line-of-sight vector $\bm{R}$ (as in \cite{Miloh1982}).
This assumption is reasonable for initial trajectory planning or for short duration collision encounters (with imminent collisions) without consideration of longer term aerodynamic or gravity effects.
The angles $\theta$, $\phi$, and $\psi$ satisfy $0 \leq \theta,\phi, \psi < \pi$ with $\psi$ also satisfy the bounds
\begin{align}
    \label{eq:psiBounds}
    \begin{cases}
    |\theta - \psi|
    \leq \psi
    \leq \theta + \phi & \text{for} \quad \theta + \phi \leq \pi\\
    |\theta - \psi|
    \leq 2\pi - \psi
    \leq \theta + \phi & \text{for} \quad \theta + \phi \geq \pi.
    \end{cases}
\end{align}

\begin{figure}[t!]
    \centering
    \includegraphics[width = 3.2in]{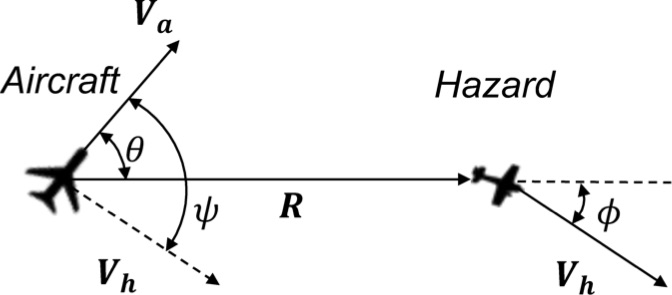}
    \caption{Range and velocity vectors of aircraft and hazard in three dimensions together with relative angles.}
    \label{fig:fig3d}
\end{figure}

Under the assumption of constant speeds $v_a = ||\bm{V}_a||$ and $v_h = ||\bm{V}_h||$, the range $r(t)$ varies according to
\begin{align}
    \label{eq:3dRangeDynamics}
    \dot{r}
    &= v_h \cos \phi - v_a \cos \theta
\end{align}
for $t \geq 0$ where $\dot{r}$ denotes the total derivative of $r(t)$ with respect to time, and here we omit the time arguments for brevity. 
We assume that the aircraft can only accelerate in a direction normal to its instantaneous velocity. 
Then, we define vectors $\bm{W}_a(t) \in \mathbb{R}^3$ along with the bound $||\bm{W}_a(t)|| \leq \bar{w}_a$ so that the aircraft's total acceleration is given by
\begin{align}
    \label{eq:aircraftDynamics}
   \dot{\bm{V}}_a
    &= \bm{W}_a \times \bm{V}_a
\end{align}
at any time $t \geq 0$.
Due to the hazard being agile, it can instantaneously change the direction of its velocity  $\bm{V}_h$, which equivalently corresponds to being able to instantaneously change the relative bearing angle $\phi(t)$ and the relative heading angle $\psi(t)$.
By differentiating \eqref{eq:bearingToHazard} and substituting both \eqref{eq:3dRangeDynamics} and \eqref{eq:aircraftDynamics}, the bearing of the hazard from the aircraft varies according to
\begin{align}
    \label{eq:3dBearingDynamics}
    \begin{aligned}
        \dot{\theta}
        &= \dfrac{1}{r \sin \theta} \left[ v_a \sin^2 (\theta) + v_h \left( \cos \phi \cos \theta - \cos \psi \right) \right]\\
        &\qquad- \sin^{-1}(\theta) \left< \bm{W}_a, \bm{e}_{\bm{V}_a} \times \bm{e}_{\bm{R}} \right>
    \end{aligned}
\end{align}
for $t \geq 0$.
The (relative) motion of the hazard and aircraft is thus described by the states $r$ and $\theta$ satisfying the state equations \eqref{eq:3dRangeDynamics} and \eqref{eq:3dBearingDynamics} with the angles $\phi$ and $\psi$ selected by the hazard, and $\bm{W}_a$ selected by the aircraft.

\subsection{Differential Game Formulation and Bearing-Only Collision Avoidance}
To formulate our collision avoidance problem, we assume that the range is initially decreasing (i.e. that the initial values of $r(0)$ and $\theta(0)$ are such that $\dot{r}(0) < 0$).
The miss-distance is the minimum range $r(T)$ between the aircraft and hazard attained at time $t = T$ when $\dot{r}(t) < 0$ for $t \in [0,T)$ and $\dot{r}(T) = 0$.
The aircraft seeks to maximise the miss-distance $r(T)$ through selection of its acceleration vectors $\bm{W}_a$ whilst the hazard seeks to minimise the miss-distance through selection of the direction of its velocity vector (i.e., the angles $\phi$ and $\psi$).
The aircraft and hazard thus play the differential game defined by the value function
\begin{align}
 \label{eq:3dProblem}
\begin{aligned}
	J(r_0,\theta_0) \triangleq \max_{\bm{W}_a} \min_{\phi,\psi} r(T)
\end{aligned}
\end{align}
where $r(T)$ is the (terminal) cost function of the game and the optimisations are subject to the angle constraints \eqref{eq:psiBounds}, the kinematic constraints \eqref{eq:3dRangeDynamics} and \eqref{eq:3dBearingDynamics}, the control constraint $||\bm{W}_a|| < \bar{w}_a$, the minimum range constraint $\dot{r}(T) = 0$, and the initial conditions $r(0) = r_0$ and $\theta(0) = \theta_0$.

In general, the aircraft strategies (and the terminal time $T$) solving \eqref{eq:3dProblem} could be functions of the four variables $(r,\theta,\phi,\psi)$ that are sufficient to describe the configuration of the game along with the speeds $v_a$ and $v_h$ (cf.~\cite[Section 8.2]{Basar1999} and its discussion of feedback strategies in two-player zero-sum differential games).
However, since the angles $\phi$ and $\psi$ are controlled directly by the agile hazard (and can be instantaneously changed), knowledge of them will intuitively not assist the aircraft in determining its optimal manoeuvres.
Furthermore, the aircraft's objective of maximising the minimum range $r(T)$ places no explicit emphasis on the absolute value of the range at any time $t \in [0,T]$, and so the aircraft is likely to be able to determine its optimal manoeuvres (and the terminal time $T$) with only knowledge of the bearing angle $\theta$ and the speeds $v_a$ and $v_h$.
To determine a bearing-only strategy for collision avoidance, we therefore seek to show that there exists an aircraft strategy with an associated terminal time $T$ (both dependent only on $\theta$) that solves \eqref{eq:3dProblem} over all possible speeds $v_h$ and $v_a$.

\section{Differential Game Solution}
\label{sec:diffGame}

In this section, we solve the differential game \eqref{eq:3dProblem} by exploiting the framework developed in \cite{Isaacs65} (see also \cite[Chapter 8]{Basar1999} and \cite{Merz1971}) and extended in \cite{Miloh1982} for differential games in three dimensions.
We specifically use the Isaacs or Hamilton-Jacobi-Isaacs (HJI) equation and an argument similar to that of \cite{Miloh1982} to reduce the game in three dimensions \eqref{eq:3dProblem} to a differential game in two dimensions.
For this two-dimensional (or planar) differential game, we determine the optimal aircraft and hazard strategies as functions of the partial derivatives of the value function by following the approach of \cite{Basar1999} (see also \cite{Merz1971}).
Using these strategies, we identify the optimal trajectories of the hazard relative to the aircraft for all regions of the game space.
%Finally, the optimal strategy for the aircraft is shown to constitute a bearing-only feedback law.

% Let us define the Hamiltonian function for the game \eqref{eq:3dProblem} as
% \begin{align*}
%     \mathcal{H} (r, \theta, \phi, \psi, W_a, \Lambda^r, \Lambda^\theta)
%     &\triangleq \Lambda^r \dot{\theta} + \Lambda^\theta \dot{\theta}
% \end{align*}
% where $\Lambda^r(t) \in \mathbb{R}$ and $\Lambda^\theta(t) \in \mathbb{R}$ are costate (or adjoint) functions.

\subsection{Reduction to Planar Problem}
We first note that the order of the minimisations and maximisations in \eqref{eq:3dProblem} is reversible since the kinematics \eqref{eq:3dRangeDynamics} and \eqref{eq:3dBearingDynamics} are separable in the control inputs, and the cost function $r(T)$ is terminal (see the minimax assumption of \cite{Isaacs65}).
The HJI equation for \eqref{eq:3dProblem} is then
\begin{align}
    \label{eq:3dHjb}
    \max_{\bm{W}_a} \min_{\phi,\psi} \left[ J_r \dot{r} + J_\theta \dot{\theta} \right]
    &= 0
\end{align}
for all $t \in [0,T]$ where $J_r$ and $J_\theta$ denote the partial derivatives of the value function $J$ with respect to its first and second arguments, respectively, evaluated at $r(t)$ and $\theta(t)$ (see \cite[Section 8.2.1]{Basar1999} for a detailed derivation of the HJI equation in zero-sum differential games).
Substituting \eqref{eq:3dRangeDynamics} and \eqref{eq:3dBearingDynamics} into \eqref{eq:3dHjb} for $\dot{r}$ and $\dot{\theta}$, we have that the optimal aircraft strategy is to select its acceleration such that
\begin{align}
    \label{eq:3dOptimalAircraftAdjoint}
    \bm{W}_a^*
    &= - \bar{w}_a \sin^{-1}(\theta) \left( \bm{e}_{\bm{V}_a} \times \bm{e}_{\bm{R}} \right) \sign \left( J_\theta \right)
\end{align}
where $\bm{W}_a^*$ denotes the optimal value of $\bm{W}_a$ and 
\begin{align*}
    \sign (J_\theta)
    &\triangleq \begin{cases}
            1 & \text{for } J_\theta > 0 \\
            -1 & \text{for } J_\theta < 0.
       \end{cases}
\end{align*}
Substituting \eqref{eq:3dRangeDynamics} and \eqref{eq:3dBearingDynamics} into \eqref{eq:3dHjb} similarly implies that the hazard's optimal strategy is to select angles $\phi$ and $\psi$ solving
\begin{align*}
    \min_{\phi,\psi} \left[ \cos \phi \left( J_r + \dfrac{1}{r} J_\theta \tan^{-1} \theta \right) + \cos \psi \left( - \dfrac{1}{r} J_\theta \sin^{-1} \theta \right) \right]
\end{align*}
subject to the bounds in \eqref{eq:psiBounds}.
As in \cite[p.~740]{Miloh1982}, the bounds on $\psi$ given by \eqref{eq:psiBounds} combined with application of Isaac's lemma on circular vectograms implies that the hazard's optimal strategy is to select $\phi^*$ and $\psi^*$ such that
\begin{align}
    \label{eq:3dOptimalHazardAdjoint}
    \tan \phi^*
    &= \pm \dfrac{J_\theta}{r J_r}
\end{align}
and either $\psi^* = \theta + \phi^*$ or $\psi^* = |\theta - \phi^*|$ hold.
We note that both \eqref{eq:3dOptimalAircraftAdjoint} and \eqref{eq:3dOptimalHazardAdjoint} hold when $\theta = 0$ since $\sin^{-1}(\theta) \left( \bm{e}_{\bm{V}_a} \times \bm{e}_{\bm{R}} \right) = 1$ when $\theta = 0$.

From \eqref{eq:3dOptimalAircraftAdjoint}, we see that the aircraft's optimal strategy is to accelerate in the plane containing $\bm{V}_a$ and $\bm{R}$ in one of two directions determined by the sign of $J_\theta$.
In view of \eqref{eq:3dOptimalHazardAdjoint} and since $\phi^* = \psi^* - \theta$ when the vectors $\bm{e}_{\bm{R}}$, $\bm{e}_{\bm{V}_a}$, and $\bm{e}_{\bm{V}_h}$ are coplanar (with the angles defined on the interval $[-\pi,\pi]$), the hazard's optimal strategy is to select its velocity $\bm{V}_h$ so that it is coplanar with the line-of-sight vector $\bm{R}$ and the aircraft's velocity $\bm{V}_a$.
Under optimal play, the aircraft and the hazard therefore immediately drive the game to a plane in which $\bm{V}_h$, $\bm{R}$ and $\bm{V}_a$ are coplanar, and neither seeks any deviation from this plane.
The game in three dimensions thus reduces to a game in two dimensions on the plane defined by $\bm{R}$ and $\bm{V}_a$ (which corresponds to the conflict plane used recently in the development of results for aircraft proximity, cf.~\cite{Fulton2018}).
We, therefore, can now proceed to characterise the two-dimensional planar game without any loss of generality.

\subsection{Planar Differential Game Formulation}

To describe the solution of the collision avoidance game in the plane defined by $\bm{R}$ and $\bm{V}_a$, let us cast the game into two dimensions.
Let us define a coordinate system with its $y$-axis aligned with the aircraft's velocity vector, and its $x$-axis completing a right-hand coordinate system (as shown in Fig.~\ref{fig:fig1}).
In this coordinate system, $\theta$ and $\psi$ are measured as positive in a clockwise direction from the $y$-axis, and we have that $\phi = \psi - \theta$.
Substitution of $\phi = \psi - \theta$ into \eqref{eq:3dRangeDynamics} and \eqref{eq:3dBearingDynamics} gives the planar equations of motion
\begin{align}
    \begin{split}
        \label{eq:planarEquationsOfMotion}
        \dot{r}(t)
        &= - \cos \theta(t) + v_h \cos \left( u_h(t) - \theta(t) \right)\\
        \dot{\theta}(t)
        &= -u_a(t) + \dfrac{1}{r(t)} \left[ \sin \theta(t) + v_h \sin \left( u_h(t) - \theta(t) \right) \right]
    \end{split}
\end{align}
where here we have also normalised the kinematics to an aircraft speed of $v_a = 1$.
We also use $u_h \triangleq \psi$ to denote the hazard's control input, and $u_a \in [-1,1]$ to denote the aircraft's planar turning acceleration (which we also normalise to a maximum value of $1$).
A value of $u_a = -1$ represents a maximum-rate left turn whilst $u_a = 1$ corresponds to a maximum-rate right turn.

\begin{figure}[t!]
    \centering
    \includegraphics[width = 3in]{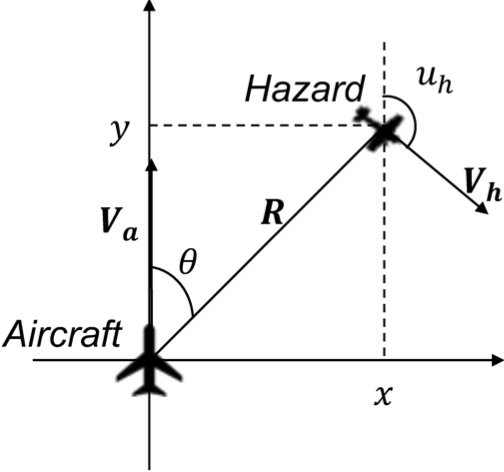}
    \caption{Coordinate system attached to the aircraft in the two-dimensional plane defined by the initial aircraft velocity vector $\bm{V}_a$ and the line-of-sight vector $\bm{R}$.}
    \label{fig:fig1}
\end{figure}

In this coordinate system, the game state reduces to the planar position of the hazard in polar form $(r,\theta)$, and the value function is
\begin{align}
 \label{eq:problemPolar}
\begin{aligned}
V (r_0,\theta_0) \triangleq \max_{u_a} \min_{u_h} r(T)
\end{aligned}
\end{align}
where $r(T)$ is the (terminal) cost function of the game and the optimisations are subject to the constraints
\begin{align*}
\begin{aligned}
& & & \dot{r}(t) = - \cos \theta(t) + v_h \cos \left( u_h(t) - \theta(t) \right), \; r(0) = r_0\\
& & & \dot{\theta}(t) = -u_a(t) + \dfrac{1}{r(t)} \left[ \sin \theta(t) + v_h \sin \left( u_h(t) - \theta(t) \right) \right]\\
& & & \theta(0) = \theta_0\\
& & & u_a(t) \in [-1,1]\\
& & & u_h(t) \in [-\pi,\pi]\\
& & & \dot{r}(T) = 0.
\end{aligned}
\end{align*}
Alternatively, by writing the hazard's position in Cartesian form with $r = \sqrt{x^2 + y^2}$, $x = r \sin \theta$, and $y = r \cos \theta$, the (planar) game can be formulated as
\begin{align}
 \label{eq:problem}
\begin{aligned}
V (x_0,y_0) \triangleq \max_{u_a} \min_{u_h} r(T) = \max_{u_a} \min_{u_h} \sqrt{x^2(T)  + y^2(T) },
\end{aligned}
\end{align}
where the optimisations are subject to the constraints
\begin{align}
\label{eq:problemConstraints}
\begin{aligned}
& & & \dot{x}(t) = -u_a(t) y(t) + v_h \sin u_h(t), \; x(0) = x_0\\
& & & \dot{y}(t) = -1 + u_a(t) x(t) + v_h \cos u_h(t), \; y(0) = y_0\\
& & & u_a(t) \in [-1,1]\\
& & & u_h(t) \in [-\pi,\pi]\\
& & & \dot{r}(T) = 0.
\end{aligned}
\end{align}

\subsection{Optimal Planar Strategies}
To solve the planar game of \eqref{eq:problemPolar} or \eqref{eq:problem}, we first note that as in \eqref{eq:3dProblem}, the order of the minimisations and maximisations in \eqref{eq:problemPolar} and \eqref{eq:problem} are reversible since the kinematics are separable in the control inputs, and the cost function $r(T)$ is terminal (see the minimax assumption of \cite{Isaacs65}).
The HJI equation for \eqref{eq:problem} is then
\begin{align}
    \label{eq:hjb}
    \max_{u_a \in [-1,1]} \min_{u_h \in [-\pi,\pi]} \left[ V_x \dot{x} + V_y \dot{y} \right]
    &= 0
\end{align}
for all $t \in [0,T]$ where $V_x$ and $V_y$ denote the partial derivatives of the value function $V$ with respect to its first and second arguments, respectively, evaluated at $x(t)$ and $y(t)$ (see again \cite[Section 8.2.1]{Basar1999} for a detailed derivation of the HJI equation).
By substituting the expressions for $\dot{x}$ and $\dot{y}$ from \eqref{eq:problemConstraints} into \eqref{eq:hjb}, we have that 
\begin{align}
    \label{eq:hamExpanded}
    0
    &= -V_y + \max_{u_a \in [-1,1]} u_a \left[ V_y x - V_x y \right] + \min_{u_h \in [-\pi,\pi]} \left[ V_x v_h \sin u_h + V_y v_h \cos u_h \right].
\end{align}
The optimal aircraft controls solving \eqref{eq:hamExpanded} therefore satisfy
\begin{align}
    \label{eq:optimalAircraftAdjoint}
    u_a^*(t)
    &= \sign \sigma(t)
\end{align}
where $\sigma(t) \triangleq  V_y x - V_x y$ is a switching function.
The aircraft will thus either turn hard right or hard left depending on the sign of the switching function $\sigma(t)$.
Application of the lemma on circular vectograms of \cite[Lemma 2.8.1]{Isaacs65} to \eqref{eq:hamExpanded} also gives that the optimal strategy for the hazard is to select $u_h$ such that the vector $(\sin u_h,\cos u_h)$ is parallel to the vector $(-V_x, -V_y)$, and so the optimal hazard controls $u_h^*$ satisfy
\begin{align}
    \label{eq:optimalHazardAdjoint}
    V_x
    = - \mu \sin u_h^*(t)
    \quad \text{and} \quad
    V_y 
    = - \mu \cos u_h^*(t)
\end{align}
where $\mu \triangleq \sqrt{V_x^2 + V_y^2} > 0$.

\subsection{Optimal Planar Trajectories}
We now use the optimal strategies described by \eqref{eq:optimalAircraftAdjoint} and \eqref{eq:optimalHazardAdjoint} to compute the optimal trajectories of the hazard.
We shall compute these trajectories by dividing the game space into ``regular'' regions and singular arcs.

% For conviencnce, let us introduce t
% .$H(x,y,u_a,u_h,V_x,V_y)$

% let us introduce the Hamiltonian function
% \begin{align}
%     \label{eq:ham}
%     H(x,y,u_a,u_h,V_x,V_y)
%     \triangleq V_x \dot{x} + V_y \dot{y}
% \end{align}

\subsubsection{Regular Regions}
In regular regions, the switching function $\sigma(t)$ has a constant sign (i.e., $\sigma(t) > 0$ or $\sigma(t) < 0$) and so the optimal aircraft controls are $u_a^* = \pm 1$.
The value function $V$ also has continuous second derivatives in regular regions so that
\begin{align*}
    \dot{V}_x
    = - \pdv{}{x} \left[ V_x \dot{x}^* + V_y \dot{y}^* \right]
    = - u_a^* V_y
\end{align*}
and
\begin{align*}
    \dot{V}_y
    = - \pdv{}{y} \left[ V_x \dot{x}^* + V_y \dot{y}^* \right]
    = u_a^* V_x
\end{align*}
where $\dot{x}^*$ and $\dot{y}^*$ denote $\dot{x}$ and $\dot{y}$, respectively, evaluated with the optimal aircraft controls $u_a^* = \pm 1$ determined by \eqref{eq:optimalAircraftAdjoint} (cf.~\cite[Section 8.6]{Basar1999}).
In \emph{retrogressive time} where $\tau \triangleq T-t$ and ${}^\circ$ denotes the retrogressive time derivatives of quantities with respect to $\tau$, we therefore have that
\begin{align}
    \label{eq:adjointRetroDiff}
    \bot{V}_x
    = u_a^* V_y
    \quad \text{and} \quad
    \bot{V}_y
    = -u_a^* V_x
\end{align}
in regular regions.
Let us denote the retrogressive time initial conditions associated with these differential equations as $\tilde{u}_a \triangleq u_a^*(T)$, $\tilde{u}_h \triangleq u_h^*(T)$, $\tilde{V}_x \triangleq V_x(x(T),y(T))$, and $\tilde{V}_y \triangleq V_y(x(T),y(T))$.
Solving the (retrogressive time) differential equations \eqref{eq:adjointRetroDiff} in regular regions from (arbitrary) retrogressive time initial conditions $\tilde{u}_a = \pm 1$, $\tilde{u}_h$, $\tilde{V}_x$, and $\tilde{V}_y$, where
\begin{align*}
    \tilde{V}_x
    = - \tilde{\mu} \sin \tilde{u}_h
    \quad \text{and} \quad
    \tilde{V}_y
    = - \tilde{\mu} \cos \tilde{u}_h
\end{align*}
due to \eqref{eq:optimalHazardAdjoint}, gives that
\begin{align}
    \label{eq:adjointSolutions}
    V_x
    = - \tilde{\mu} \sin \left( \tilde{u}_h + \tau u_a^* \right)
    \, \text{and} \,
    V_y
    &= - \tilde{\mu} \cos \left( \tilde{u}_h + \tau u_a^* \right)
\end{align}
with $\tilde{\mu} \triangleq \sqrt{\tilde{V}_x^2 + \tilde{V}_y^2}$.
Comparing these solutions to \eqref{eq:optimalHazardAdjoint} implies that in regular regions the hazard's optimal strategy is to vary the relative heading $u_h$ piecewise linearly in time according to
\begin{align}
    \label{eq:optimalHazardControlViaAdjoint}
    u_h^*(\tau)
    &= \tilde{u}_h + \tau u_a^*.
\end{align}
With this expression for the optimal hazard strategy in regular regions, the equations of motion in \eqref{eq:problemConstraints} can be solved analytically in retrogressive time from arbitrary retrogressive time initial conditions $(\tilde{x},\tilde{y}) \triangleq (x(T),y(T))$.
Solving the equations of motion in \eqref{eq:problemConstraints} analytically in retrogressive time (with extensive algebraic manipulations and use of trigonometric identities), we obtain
\begin{align}
    \label{eq:stateSolutions}
    \begin{aligned}
        x(\tau)
        &= u_a^*(1-\cos \tau) + \tilde{x} \cos \tau + u_a^* \tilde{y} \sin \tau - v_h \tau \sin u_h^* \\
        y(\tau)
        &= (1-u_a^* \tilde{x}) \sin\tau + \tilde{y} \cos \tau - v_h \tau \cos u_h^*
    \end{aligned}
\end{align}
as the solutions to the game in regular regions where $u_h^*$ is given by \eqref{eq:optimalHazardControlViaAdjoint}, and $u_a^*$ is determined by \eqref{eq:optimalAircraftAdjoint} and \eqref{eq:optimalHazardAdjoint} in terms of the gradients $V_x$ and $V_y$ from \eqref{eq:adjointSolutions}.
To evaluate these solutions to the game in regular regions, we must now find the (optimal) terminal conditions $\tilde{V}_x$ and $\tilde{V}_y$, together with the (optimal) terminal controls $\tilde{u}_h = u_h^*(T)$ and $\tilde{u}_a = u_a^*(T)$.

To obtain the terminal conditions $\tilde{V}_x$ and $\tilde{V}_y$, it is convenient to consider the polar form of the game given in \eqref{eq:problemPolar}. 
Let us denote the partial derivatives of the value function $V$ with respect to $r$ and $\theta$ and evaluated at $r(t)$ and $\theta(t)$ as $V_r$ and $V_\theta$, respectively.
By computing the total derivatives of $V$ (in its polar form) with respect to $x$ and $y$, we have that
\begin{align*}
    V_x
    &= V_r \pdv{r}{x} + V_\theta \pdv{\theta}{x}
\end{align*}
and 
\begin{align*}
    V_y
    &= V_r \pdv{r}{y} + V_\theta \pdv{\theta}{y}
\end{align*}
at any $t \in [0,T]$.
The value function $V$ at $t = T$ is $V = r(T)$ and so $V_r  = 1 > 0$ and $V_\theta = 0$ at $t = T$.
Thus at $t = T$, we have the terminal conditions
\begin{align}
    \label{eq:VxTerminalValue}
    \tilde{V}_x
    = V_r \sin \theta(T)
    = \sin \theta(T)
\end{align}
and
\begin{align}
    \label{eq:VyTerminalValue}
    \tilde{V}_y
    = V_r \cos \theta(T)
    = \cos \theta(T).
\end{align}
To obtain the terminal hazard controls $\tilde{u}_h = u_h^*(T)$, we note that the HJI equation \eqref{eq:hjb} for the polar form of the game \eqref{eq:problemPolar} at $t = T$ is
\begin{align*}
    0
    &= \max_{u_a} \min_{u_h} \left[ V_r \dot{r} + V_\theta \dot{\theta} \right]\\
    &= \max_{u_a} \min_{u_h} \dot{r}(T)\\
    &= \min_{u_h} \left[ - \cos \theta(T) + v_h \cos \left( u_h(T) - \theta(T) \right) \right]
\end{align*}
where the second line holds due to $V_r = 1$ and $V_\theta = 0$ at $t = T$ and the last line follows by substituting the definition of $\dot{r}$. 
The minimising terminal hazard control $\tilde{u}_h = u_h^*(T)$ is then
\begin{align}
    \label{eq:hazardTerminalControl}
    u_h^*(T)
    = \theta(T) + \pi
\end{align}
which implies that
\begin{align}
    \label{eq:terminationLine}
	\dot{r}(T)
	= - \cos \theta(T) - v_h
	= 0.
\end{align}
The game therefore terminates with the hazard somewhere along the line $\cos \theta(T) = - v_h$, and pointing directly towards the aircraft.

To determine the aircraft's terminal controls $\tilde{u}_a = u_a^*(T)$, we consider the Cartesian form of the game given in \eqref{eq:problem} and note that by computing the total derivative of $V$ with respect to $\theta$ (in its Cartesian form) we have that
\begin{align*}
    V_\theta
    &= V_x \pdv{x}{\theta} + V_y \pdv{y}{\theta}\\
    &= V_x r \cos \theta - V_y r \sin \theta\\
    &= V_x y - V_y x\\
    &= - \sigma(t).
\end{align*}
We therefore have that $\sigma(T) = 0$ since $V = r(T)$ at $t = T$ implies that $V_\theta = 0$.
Thus, the terminal aircraft controls cannot be determined directly with \eqref{eq:optimalAircraftAdjoint}.
Instead, let us consider
\begin{align*}
    \dv{\sigma(T)}{t}
    = \tilde{V}_x
    = \sin \theta(T)
\end{align*}
where we have used that $u_h^*(T)$ is given by \eqref{eq:hazardTerminalControl} along with the terminal value $\tilde{V}_x$ from \eqref{eq:VxTerminalValue}.
Since $\sigma(T) = 0$, the switching function will satisfy $\sigma(t) > 0$ for $t$ near $T$ when $\dv{\sigma(T)}{t} < 0$ and $\sigma(t) < 0$ for $t$ near $T$ when $\dv{\sigma(T)}{t} > 0$.
For $\theta \in (-\pi, \pi]$, $\dv{\sigma(T)}{t} = \sin \theta(T)$ is negative when $\theta(T) < 0$ and positive when $\theta(T) > 0$.
Thus, the aircraft's controls at $t$ near $T$ satisfy
\begin{align}
    \label{eq:aircraftTerminalControl}
    u_a^*
    &= -\sign \theta(T)
\end{align}
and so when the game terminates with the hazard in the right half plane with $\theta(T) = \acos (-v_h)$, the aircraft is turning left away from the hazard. 
When the game terminates with the hazard in the left half plane with $\theta(T) = -\acos (-v_h)$, the aircraft is turning right away from the hazard. 

With the terminal controls $u_a^*$ and $u_h^*$ determined by \eqref{eq:hazardTerminalControl} and \eqref{eq:aircraftTerminalControl}, and the terminal values of $V_x$ and $V_y$ determined by \eqref{eq:VxTerminalValue} and \eqref{eq:VyTerminalValue}, we can evaluate the equations \eqref{eq:adjointSolutions} -- \eqref{eq:stateSolutions} at any retrogressive time $\tau$. 
These equations are valid from $\tau = 0$ for states $(\tilde{x},\tilde{y})$ on the line of minimum range defined by \eqref{eq:terminationLine} until the sign of $\sigma$ (and hence $u_a^*$) changes.
As illustrated in Fig.~\ref{fig:fig3}, the resulting state trajectories $(x,y)$ fill the game space and describe optimal paths of the hazard in the coordinate system of Fig.~\ref{fig:fig1}.
Along the optimal trajectories, the aircraft does not switch between $u_a^* \pm 1$ except across the $y$-axis.
In standard time $t$, the optimal trajectories begin in the game space and terminate at the lines of minimum range defined by
\begin{align}
    \label{eq:lineOfMinimumRange}
    \cos \theta = - v_h.
\end{align}
Clearly, the lines of minimum range \eqref{eq:lineOfMinimumRange} are only defined when $0 < v_h \leq v_r = 1$, and so the game can only terminate with $r(T) > 0$ when the hazard is slower than the aircraft.
If the hazard is faster than the aircraft (i.e., if $v_h > 1$), then the game terminates at the origin with $r(T) = 0$.

\begin{figure}[t!]
    \centering
    \includegraphics[width = 3in]{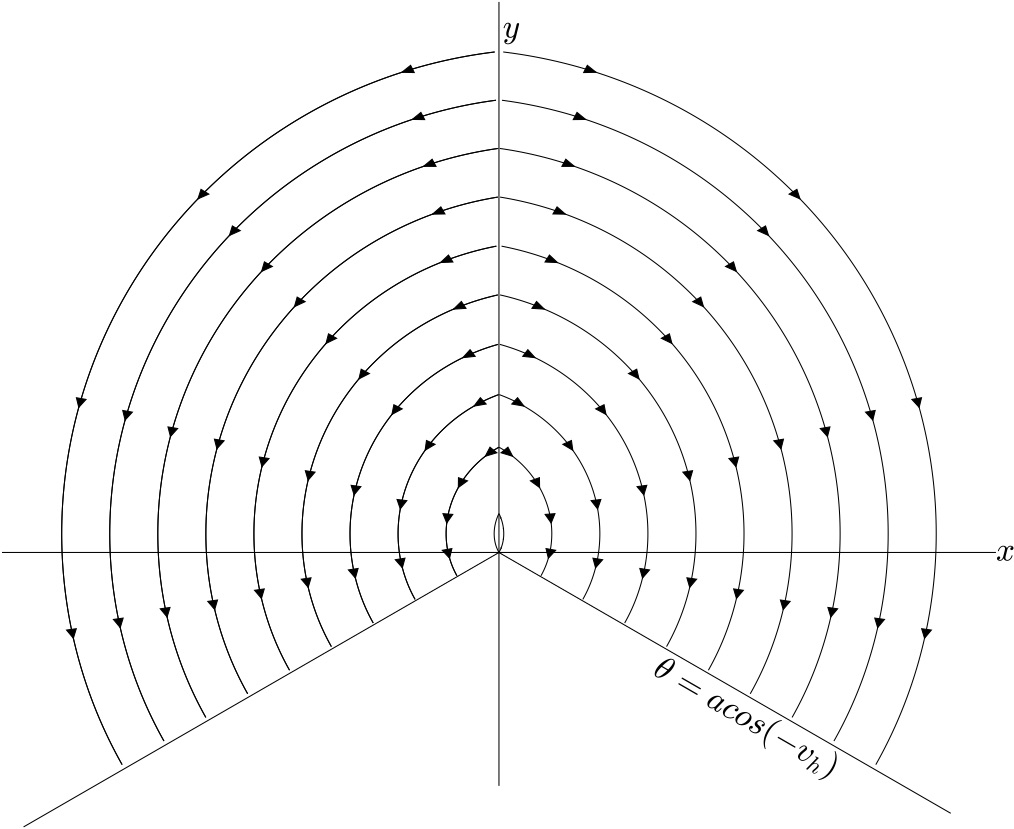}
    \caption{Optimal trajectories of the hazard in the coordinate system of Fig.~\ref{fig:fig1} terminating (forward in time) on the lines of minimum range defined by $\cos \theta = - v_h$. Illustrated trajectories are for $v_h = 0.5$ but are representative of trajectories for all $0 < v_h < 1$.}
    \label{fig:fig3}
\end{figure}

\subsubsection{Singular Arcs}
For completeness, we now check for the existence of singular arcs.
Along singular arcs, the switching function $\sigma(t)$ and its time-derivatives vanish, namely, 
\begin{align}
    \label{eq:singluarSigma}
    \sigma(t)
    &=  V_y x - V_x y= 0
\end{align}
and 
\begin{align}
    \label{eq:singluarSigmaDot}
    \dot{\sigma}(t)
    %&= \dot{V}_x y + V_x \dot{y} - \dot{V}_y x - V_y \dot{x}\\
    &= V_x (v_h \cos u_h^* - 1) - V_y v_h \sin u_h^* = 0.
\end{align}
Together \eqref{eq:singluarSigma} and \eqref{eq:singluarSigmaDot} define the (homogeneous) system of linear equations
\begin{align}
    \label{eq:linearSystem}
    \begin{bmatrix}
    -y & x\\
    v_h \cos u_h^*-1 & -v_h \sin u_h^*
    \end{bmatrix}
    \begin{bmatrix}
    V_x\\
    V_y
    \end{bmatrix}
    &= \begin{bmatrix}
    0\\
    0
    \end{bmatrix}.
\end{align}
We note that $V_x$ and $V_y$ cannot be simultaneously zero due to their expressions in \eqref{eq:optimalHazardAdjoint}.
The coefficient matrix in the system of linear equations \eqref{eq:linearSystem} must therefore be singular with its determinant given by
\begin{align*}
    0
    &= x - v_h \left( x \cos u_h^* - y \sin u_h^* \right)\\
    &= x + v_h \left( x V_y - V_x y \right)/\mu\\
    &= x
\end{align*}
where the second equality follows from \eqref{eq:optimalHazardAdjoint} and the last equality follows from \eqref{eq:singluarSigma}.
The only singular arc in the game therefore occurs when $x(t) = 0$.
The game always terminates at or before $x = 0$ and $y \leq 0$ and so we need only resolve the controls when $x = 0$ and $y > 0$.
When $x = 0$ and $y >0$, it is straightforward to observe that the hazard's optimal strategy is simply to point at the aircraft with $u_h^* = \pi$, and the aircraft's optimal strategy is nonunique with $u_a^* = \pm 1$ being equivalent.
%Furthermore, the line $x = 0$ is a universal line (cf. \cite[p. 366]{Basar1999}) and so both $V$, $V_x$, and $V_y$ are continuous across it.
%Trajectories can therefore leave the $y$-axis $x = 0$ in retrogressive time at any instant $\tau_s$ and at any $y = \tilde{y} > 0$.
%For the trajectory leaving the point $(0,\tilde{y})$ with either $u_a = \pm 1$, we have that $u_h(\tau_s) = 0$ and so as in \eqref{eq:optimalHazardControlViaAdjoint}, $u_h(\tau) = \tau u_a$ and 
%\begin{align}
%    \label{eq:adjointSolutionsSingular}
%    V_x
%    = \tilde{\mu} \sin \left( \tau u_a \right)
%    \quad \text{and} \quad 
%    V_y
%    = \tilde{\mu} \cos \left( \tau u_a \right)
%\end{align}
%with $\tilde{\mu} = 1/(1-v_h)$ by equating \eqref{eq:hjb} at $\tau = \tau_s$ noting that at $\tau = \tau_s$ $x = V_x = u_h = 0$.
%With the values of $u_a$, $u_h$, $V_x$, and $V_y$ at $\tau_s$, the state equations \eqref{eq:stateSolutions}, hazard controls $u_h(\tau) = \tau u_a$, and value function gradients \eqref{eq:adjointSolutionsSingular} can be solved backwards in time from points $(0,\tilde{y})$ on the singular surface until the sign of $\sigma(t)$ (and hence $u_a$) changes or the trajectory leaves the game space of valid initial conditions.

\subsection{Optimal Aircraft and Hazard Strategies}
The reduction of the three-dimensional game \eqref{eq:3dProblem} to a game on the plane defined by $\bm{V}_a$ and $\bm{R}$ combined with our analysis of the optimal trajectories of the planar game in Fig.~\ref{fig:fig3} implies that the aircraft's optimal strategy from any region of the game space is to turn away from the hazard in the plane defined by $\bm{V}_a$ and $\bm{R}$ until the relative bearing $\theta$ satisfies $\cos \theta = - v_h$ (or $r(T) = 0$ as in the case $v_h > v_a = 1$).
%From \eqref{eq:optimalHazardControlViaAdjoint} and \eqref{eq:hazardTerminalControl}, we see that the hazard's optimal response is to vary its heading so as to always point directly at the aircraft (i.e., to follow a pure pursuit strategy).
The aircraft's optimal strategy solving the collision avoidance game \eqref{eq:3dProblem} can therefore be summarised as
\begin{align}
    \label{eq:aircraftRule}
    u_a^*(t)
    = \begin{cases}
            -1  & \text{for } \theta(t) \in (0,\acos(-v_h))\\
            1 & \text{for } \theta(t) \in (-\acos(-v_h),0)\\
            \pm 1 & \text{for } \theta(t) = 0
        \end{cases}
\end{align}
where the turn is in the plane defined by $\bm{V}_a$ and $\bm{R}$, and corresponds to employing the maximum acceleration magnitude of $\bar{w}_a$.
%If the hazard is slower than the aircraft (i.e., $0 < v_h \leq 1$), then turning according to \eqref{eq:aircraftRule} can ensure that $r(T) > 0$.
%If the hazard is faster than the aircraft (i.e., $v_h > 1$), then $r(T) = 0$ and the aircraft can only delay the time until this occurs by turning according to \eqref{eq:aircraftRule}.

The hazard's optimal response is described by \eqref{eq:optimalHazardControlViaAdjoint} and involves varying the relative heading $u_h^*$ in the plane defined by $\bm{V}_a$ and $\bm{R}$ linearly in time with a constant rate of change determined by (the non-switching) aircraft control $u_a^* = \pm 1$.
As illustrated in Fig.~\ref{fig:fig3}, the optimal trajectories of the hazard in the (planar) coordinate system relative to the aircraft are therefore curved.
However, since the change in the relative heading $u_h^*$ is entirely due to the turning of the aircraft, in an inertial frame with reference point fixed at some arbitrary point in the two-dimensional plane defined by $\bm{V}_a$ and $\bm{R}$, the hazard's optimal trajectory is a straight line.
As established in \eqref{eq:hazardTerminalControl}, the game terminates with the hazard pointing directly towards the aircraft, and so the line that the hazard follows is the line-of-sight vector $\bm{R}(T)$ at the terminal time $T$.
When $v_h < 1$, the line-of-sight vector $\bm{R}(T)$ is therefore aligned with the line of minimum range defined by \eqref{eq:lineOfMinimumRange}.
The hazard's optimal strategy is illustrated in Fig.~\ref{fig:inertialTrajectories} for $v_h = 0.5 < v_a = 1$.

\begin{figure}[t!]
    \centering
    \begin{subfigure}[b]{3in}
        \includegraphics[width=\textwidth]{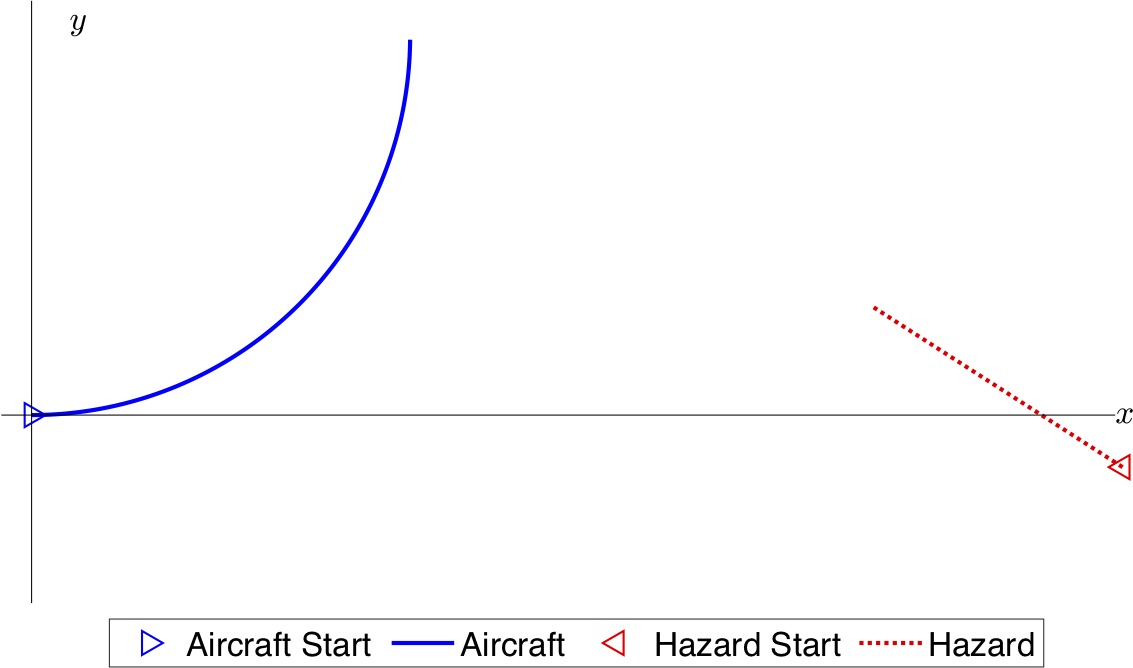}
        \caption{}
    \end{subfigure}
    \hfil
    \begin{subfigure}[b]{1.2in}
        \includegraphics[width=\textwidth]{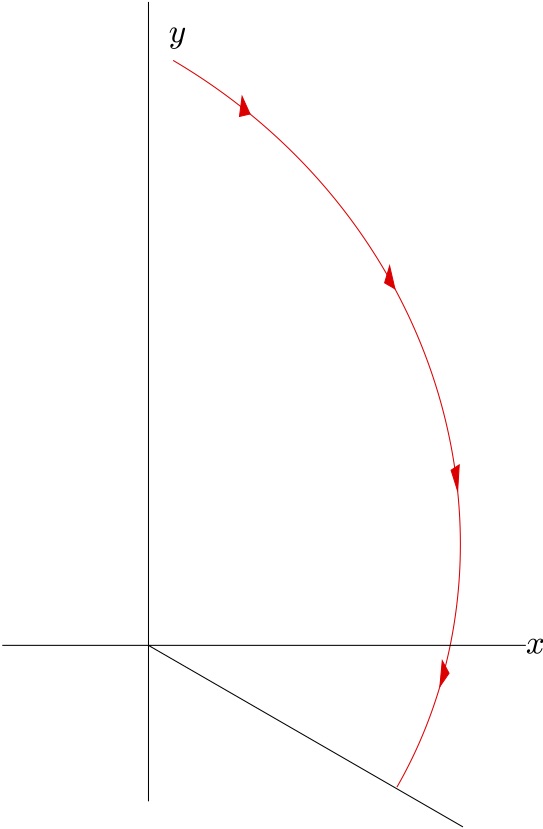}
        \caption{}
    \end{subfigure}
    \caption{Optimal trajectories for $v_h = 0.5$ on plane defined by $\bm{V}_a$ and $\bm{R}$ of (a) Aircraft and Hazard in fixed inertial coordinate system, and (b) Hazard in the relative coordinate system of Fig.~\ref{fig:fig1}. Trajectories terminate at the minimum range $r(T)$.}
    \label{fig:inertialTrajectories}
\end{figure}

\section{Further Properties of the Game Solution and the Bearing-Only Avoidance Strategy}
\label{sec:properties}

In this section, we examine further properties of the solution to the collision avoidance differential game \eqref{eq:problem}.
We first examine the sets of states that will lead to collision when both the aircraft and hazard play optimally, and we compare our results with the pursuit-evasion game solution of \cite{Exarchos2015}.
We then show that the optimal aircraft strategy \eqref{eq:aircraftRule} can be recast as a pure bearing-only strategy (without requiring knowledge of $v_h$), and can be extended to solving the problem of regaining a desired separation after maximising the miss-distance.
We further examine the optimality and robustness of this bearing-only strategy in cases where the hazard motion is neglected (i.e., when the hazard is  stationary) and in cases where the hazard, like the aircraft, has a finite turn rate.

\subsection{States Leading to Collision and Relation to Pursuit-Evasion Games}

Although we have found a strategy that enables the aircraft to maximise the minimum range $r(T)$, we have not quantified this range nor discussed its relationship to the initial state of the game.
Quantifying the miss-distance $r(T)$ is of significance since collisions or separation violations are often defined as occurring when $r(T)$ is below some threshold $\rho > 0$ (e.g., \cite{Exarchos2016}).
By recalling our discussion of the lines of minimum range defined by \eqref{eq:lineOfMinimumRange}, we note that $r(T) = 0$ from any initial state when the aircraft is slower than the hazard.
When the aircraft is faster than the hazard however, inspection of the optimal trajectories in Fig.~\ref{fig:fig3} suggests that the barrier separating initial states that lead to $r(T) < \rho$ from states with $r(T) > \rho$ is found by solving the state equations \eqref{eq:stateSolutions}, hazard controls \eqref{eq:optimalHazardControlViaAdjoint}, and value function gradients \eqref{eq:adjointSolutions} from the terminal states $(r(T), \theta(T)) = (\rho, \pm \acos(-v_h))$ on the lines of minimum range with $u_a^* = \pm 1$.
This barrier is illustrated in Fig.~\ref{fig:fig4}.

\begin{figure}[t!]
    \centering
    \includegraphics[width = 3in]{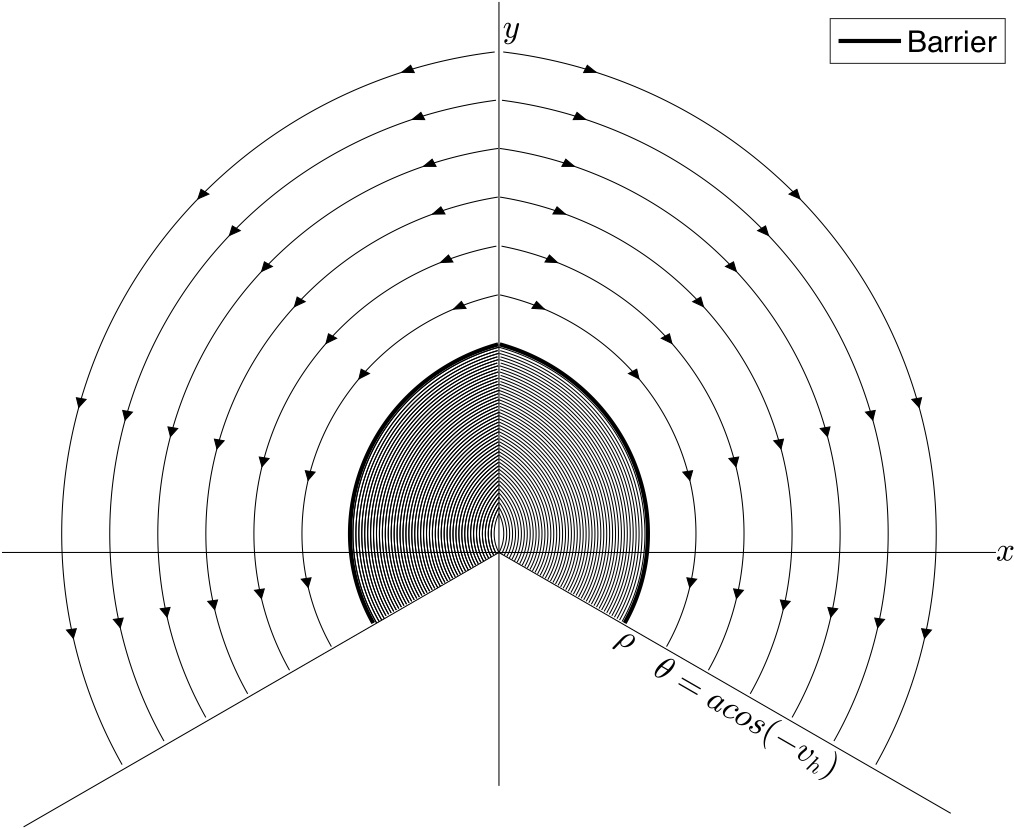}
    \caption{Optimal trajectories of the hazard in the coordinate system of Fig.~\ref{fig:fig1} with the states from which the miss-distance $r(T)$ is less than some distance $\rho > 0$ shaded. The barrier divides the states leading to $r(T) < \rho$ from those those leading to $r(T) > \rho$.}%The illustrated trajectories are for $v_h = 0.5$ and $\rho = 3.3$ but are representative of trajectories for all $0 < v_h < 1$ and $\rho > 0$.}
    \label{fig:fig4}
\end{figure}

The barrier separating the states leading to $r(T) < \rho$ from those leading to $r(T) > \rho$ provides a clear link between our planar collision avoidance game \eqref{eq:problem} and the two-dimensional pursuit-evasion game considered in \cite{Exarchos2016} (see also \cite{Exarchos2014,Exarchos2015}).
The game considered in \cite{Exarchos2016} is posed and solved with the same equations of motion as our planar game \eqref{eq:problem} but is treated as a \emph{game of kind} in which the aircraft seeks to ensure that $r(t) > \rho$ for all $t \geq 0$ whilst the hazard seeks to achieve $r(t) < \rho$ for any $t \geq 0$, for a given capture radius $\rho > 0$.
The analysis of \cite{Exarchos2016} (and \cite{Exarchos2014,Exarchos2015}) is therefore primarily concerned with determining the barrier separating the states leading to the aircraft being ``captured'' from the states leading to the aircraft evading capture indefinitely.
The barrier we have identified for the collision avoidance game \eqref{eq:problem} separating states from which $r(T) < \rho$ from those leading to $r(T) > \rho$ (illustrated in Fig.~\ref{fig:fig4}) is identical to the capture/escape barriers identified in \cite{Exarchos2015,Exarchos2014,Exarchos2016} (see for example, \cite[Fig.~4]{Exarchos2015}).

In addition to offering a new alternative derivation of the results of \cite{Exarchos2015,Exarchos2014,Exarchos2016}, our results provide new insights important for collision avoidance.
In contrast to \cite{Exarchos2015,Exarchos2014,Exarchos2016}, our planar game \eqref{eq:problem} is a \emph{game of degree} in which the miss-distance is the payoff.
We have therefore identified the lines of minimum range \eqref{eq:lineOfMinimumRange} that appear as termination conditions in the optimal aircraft strategy \eqref{eq:aircraftRule} and describe the miss-distance $r(T)$ from different initial states. 
The results and strategies established in \cite{Exarchos2015,Exarchos2014,Exarchos2016} are in contrast terminated by conditions for capture, which can be met before the minimum range $r(T)$ is reached, leaving the optimal collision avoidance manoeuvres after breaching the (arbitrary) radius $\rho$ undefined.
Our results also explicitly relate to three dimensions whilst those of \cite{Exarchos2015,Exarchos2014,Exarchos2016} are established only in two dimensions, and as we shall now show, the optimal strategy we have identified for the aircraft \eqref{eq:aircraftRule} can be viewed as an optimal bearing-only collision avoidance strategy.

\subsection{Optimal Bearing-Only Strategy and Extension for Regaining Separation}

From \eqref{eq:aircraftRule}, we see that the aircraft's optimal strategy from any region of the game space is to turn away from the hazard in the plane defined by $\bm{V}_a$ and $\bm{R}$ until the relative bearing $\theta$ satisfies $\cos \theta = - v_h$.
%From \eqref{eq:optimalHazardControlViaAdjoint} and \eqref{eq:hazardTerminalControl}, we see that the hazard's optimal response is to vary its heading so as to always point directly at the aircraft (i.e., to follow a pure pursuit strategy).
Implementation of the aircraft's optimal strategy is therefore dependent on the availability of the (instantaneous) bearing angle $\theta(t)$, and the hazard's speed $v_h$ (since the plane defined by $\bm{V}_a$ and $\bm{R}$ can be equivalently defined with $\bm{V}_a$ and $\theta$).
The hazard's speed $v_h$ is however only used to determine when the game terminates, and the aircraft can maximise the miss-distance across all $0 < v_h \leq 1$ by simply continuing to turn until $\theta = \pm \pi$.
Without the hazard's speed, the aircraft's strategy solving the collision avoidance game for all $0 < v_h \leq 1$ can be written as the bearing-only strategy
\begin{align}
    \label{eq:preMainResult}
    u_a(t)
    = \begin{cases}
            -1  & \text{for } \theta(t) \in (0,\pi)\\
            1 & \text{for } \theta(t) \in (-\pi,0)\\
            \pm 1 & \text{for } \theta(t) = 0.
        \end{cases}
\end{align}

Collision avoidance typically also requires a minimum separation distance to be reached before an encounter is resolved.
It is therefore desirable to extend the bearing-only strategy of \eqref{eq:preMainResult} to ensure that the range $r(t)$ is increasing for times $t \in (T,\infty)$ that are not explicitly considered in the formulation of the game $\eqref{eq:problem}$ since they occur after the minimum range $r(T)$.
Extension of the bearing-only strategy of \eqref{eq:preMainResult} to times $t \in (T,\infty)$ so as to increase the range $r(t)$ is straightforward since once $\theta = \pm \pi$, the aircraft and hazard velocity vectors are aligned and the aircraft can increase the range $r(t)$ by controlling the instantaneous range rate $\dot{r}(t)$.
Intuitively, the instantaneous range rate $\dot{r}(t)$ is maximised by selecting $u_a(t) = 0$ and its maximisation ensures that the range $r(t)$ increases if $0 < v_h < 1$, remains constant if $v_h = 1$, or decreases at its slowest possible rate if $v_h > v_r = 1$.
The extended bearing-only strategy for collision avoidance and regaining separation is thus
\begin{align}
    \label{eq:mainResult}
    u_a(t)
    = \begin{cases}
            -1  & \text{for } \theta(t) \in (0,\pi)\\
            1 & \text{for } \theta(t) \in (-\pi,0)\\
            \pm 1 & \text{for } \theta(t) = 0\\
            0 & \text{for } \theta(t) = \pm \pi.
        \end{cases}
\end{align}
For the game \eqref{eq:problem} with an agile hazard, this strategy will maximise the minimum range $r(T)$ when $v_h \leq v_a$, and maximise the time it takes before $r(T) = 0$ when $v_h > v_a$.

Whilst we have abstracted the aircraft controls as $u_a = 0$ or $u_a = \pm 1$ in the presentation of the the bearing-only strategy \eqref{eq:mainResult} (consistent with previous works on optimal collision avoidance, cf.~\cite{Tarnopolskaya2009,Tarnopolskaya2010,Exarchos2016,Maurer2012,Miloh1976,Merz1972,Merz1973,Mylvaganam2017}), in practical implementations these abstract commands will need to be mapped from the plane defined by $\bm{V}_a$ and $\bm{R}$ to vertical and horizontal aircraft manoeuvres and lower-level control commands.
A variety of mappings from the abstract controls in \eqref{eq:mainResult} to vertical and horizontal aircraft manoeuvres appear possible and will be aircraft-specific due to differing aircraft performance characteristics.
The simplest mapping we anticipate as feasible for most aircraft is to simply perform maximum rate turns and/or climbs or dives away from the hazard (and to fly straight and level away from it when it is directly behind).
More complicated mappings may be possible in order to accommodate aircraft with unequal climb and dive rates however any differences are likely to only have measurable impacts on time scales beyond those of concern in close proximity collision encounters (see for example the discussion in \cite{Maurer2012} and references therein).

We shall next explore the properties of the bearing-only strategy \eqref{eq:mainResult} for stationary hazards and hazards with finite turn rates.

\subsection{Further Properties of Bearing-Only Strategy for Stationary Hazards and Hazards with Finite Turn Rates}

To explore the properties of the bearing-only strategy \eqref{eq:mainResult} for stationary hazards and hazards with finite turn rates, let us consider the two-dimensional plane defined by $\bm{V}_a$ and $\bm{R}$ together with a (fixed inertial) planar coordinate system with its origin fixed at some arbitrary point in the plane.
The equations describing the motion of the aircraft in this fixed coordinate frame are similar to those in Isaac's Game of Two Cars or Homicidal Chauffeur games (cf.~\cite{Basar1999}), namely, 
\begin{align}
    \label{eq:aircraftDynamicsInertial}
    \begin{split}
        \dot{x}_a(t)
        &= v_a \sin \psi_a(t)\\
        \dot{y}_a(t)
        &= v_a \cos \psi_a(t)\\
        \dot{\psi}_a(t)
        &= \omega_a u_a(t)
    \end{split}
\end{align}
where $(x_a,y_a) \in \mathbb{R}^2$ is the aircraft's position, $\psi_a(t) \in (-\pi,\pi]$ is the aircraft's heading angle, and $u_a \in [-1,1]$ is the aircraft's control input as in \eqref{eq:problemPolar} and \eqref{eq:problem}.
Here, again without loss of generality, we normalise to a speed of $v_a = 1$ and a maximum turn rate of $\omega_a = 1$.
We first examine the performance bearing-only strategy \eqref{eq:mainResult} against stationary hazards.

\begin{proposition}
\label{proposition:rangeAccelerationStationary}
Consider the planar setting described in Fig.~\ref{fig:fig1} with a stationary hazard (i.e., $v_h = 0$).
Then, the bearing-only strategy \eqref{eq:mainResult} maximises the instantaneous range acceleration $\ddot{r}(t)$ for any $t \geq 0$.
% Similarly, when the hazard motion is neglected (i.e., with $v_h = 0$), the bearing-only strategy of \eqref{eq:mainResult} solves the optimal control problem
% \begin{align*}
%     \max_{u_a \in [-1,1]} r(T)
% \end{align*}
% subject to
% \begin{align*}
% \begin{aligned}
% & & & \dot{x}(t) = -u_a(t) y(t), \; x(0) = x_0\\
% & & & \dot{y}(t) = -1 + u_a(t) x(t), \; y(0) = y_0\\
% & & & \dot{r}(T) = 0.
% \end{aligned}
% \end{align*}
\end{proposition}
\begin{proof}
Consider the (fixed inertial) planar coordinate system in which the aircraft moves with kinematics \eqref{eq:aircraftDynamicsInertial}.
Without loss of generality, we shall let the origin of this planar coordinate system be the (fixed) location of the hazard.
The position of the aircraft in polar coordinates is then described by
\begin{align*}
    \dot{r}(t)
    &= \cos \left( \psi_a(t) - \theta_a(t) \right)\\
    \dot{\theta}_a(t)
    &= \dfrac{1}{r(t)} \sin \left( \psi_a(t) - \theta_a(t) \right)
\end{align*}
where $r$ is the range between the hazard and the aircraft, and $\theta_a$ is the relative bearing of the aircraft from the hazard.
The aircraft heading $\psi_a$ evolves according to
\begin{align*}
    \dot{\psi}_a(t)
    &= u_a(t).
\end{align*}

The time derivative of $\dot{r}$ is
\begin{align*}
    \ddot{r}(t) 
    &= \left( \dot{\theta}_a(t) - \dot{\psi}_a(t) \right) \sin \theta_a(t)
\end{align*}
and by recalling that $\dot{\psi}_a(t) = u_a(t)$ we have that
\begin{align*}
    \argmax_{u_a(t) \in [-1,1]} \ddot{r}(t) 
    &= - \argmax_{u_a(t) \in [-1,1]} u_a(t) \sin \theta_a(t).
\end{align*}
Thus, when $\theta_a(t) \in (0,\pi)$, the control maximising the instantaneous range acceleration $\ddot{r}(t)$ is $u_a = -1$ whilst the control maximising $\ddot{r}(t)$ when $\theta_a(t) \in (-\pi,0)$ is $u_a = 1$, the controls $u_a = \pm 1$ are equivalent when $\theta_a(t) = 0$.
The proof is complete.
\end{proof}

Proposition \ref{proposition:rangeAccelerationStationary} is intuitive given that the bearing-only strategy of \eqref{eq:mainResult} prescribes manoeuvres away from the hazard without regard for its velocity vector.
Furthermore, since range acceleration is integrable twice to range in the case of stationary hazards, Proposition \ref{proposition:rangeAccelerationStationary} intuitively suggests that the bearing-only strategy of \eqref{eq:mainResult} is optimal for maximising the minimum range of the aircraft from a stationary hazard (i.e., solving a problem analogous to \eqref{eq:problemPolar} or \eqref{eq:problem} with $v_h = 0$).

We now examine the properties of the bearing-only strategy \eqref{eq:mainResult} in the (more realistic) case where the hazard has a finite turn rate and is not capable of instantaneous heading changes.
This situation has dominated previous considerations of optimal (ship and aircraft) collision avoidance (cf.~\cite{Olsder1978,Miloh1976,Tarnopolskaya2009}).
In the following proposition, we show that the bearing-only strategy \eqref{eq:mainResult} maximises the instantaneous range acceleration $\ddot{r}(t)$ in this case. 

\begin{proposition}
\label{proposition:rangeAccelerationTwoCars}
Consider the planar setting described in Fig.~\ref{fig:fig1} with a hazard limited to motion described by
\begin{align}
    \label{eq:hazardTurnDynamics}
    \begin{aligned}
    \dot{x}_h(t)
    &= v_h \sin \theta_h(t)\\
    \dot{y}_h(t)
    &= v_h \cos \theta_h(t)\\
    \dot{\theta}_h(t)
    &= \omega_h \check{u}_h(t)
    \end{aligned}
\end{align}
where $\check{u}_h \in [-1,1]$ is the hazard's turn direction and $\omega_h$ is the hazard's (finite) maximum turn rate.
Then the bearing-only strategy \eqref{eq:mainResult} for the aircraft maximises the instantaneous range acceleration $\ddot{r}(t)$ for any $t \geq 0$.
\end{proposition}
\begin{proof}
With the hazard equations of motion \eqref{eq:hazardTurnDynamics}, the equations of motion \eqref{eq:planarEquationsOfMotion} become
\begin{align*}
    \begin{aligned}
    \dot{r}(t)
    &= - \cos \theta(t) + v_h \cos \left( \psi(t) - \theta(t) \right)\\
    \dot{\theta}(t)
    &= -u_a(t) + \dfrac{1}{r(t)} \left[ \sin \theta(t) + v_h \sin \left( \psi(t) - \theta(t) \right) \right]\\
    \dot{\psi}(t)
    &= -u_a(t) + \omega_h \check{u}_h(t)
    \end{aligned}
\end{align*}
where $\psi$ is the angle between the hazard's velocity vector and the aircraft's velocity vector.
We have again also normalised the aircraft's speed and turn rate so that $\omega_a = 1$ and $v_a = 1$ (see \cite[Section 2]{Tarnopolskaya2009} for further details).
The time derivative of $\dot{r}$ is then
\begin{align*}
    \ddot{r}(t)
    &= \dot{\theta} \sin \theta - v_h \left( \dot{\psi} - \dot{\theta} \right) \sin \left( \psi - \theta \right)
\end{align*}
and substitution of $\dot{\theta}$ and $\dot{\psi}$ into this expression gives that
\begin{align*}
    \argmax_{u_a(t) \in [-1,1]} \ddot{r}(t) 
    &= - \argmax_{u_a(t) \in [-1,1]} u_a(t) \sin \theta (t).
\end{align*}
Thus, when $\theta(t) \in (0,\pi)$, the control maximising the instantaneous range acceleration $\ddot{r}(t)$ is $u_a = -1$ whilst the control maximising $\ddot{r}(t)$ when $\theta(t) \in (-\pi,0)$ is $u_a = 1$, the controls $u_a = \pm 1$ are equivalent when $\theta(t) = 0$.
The proof is complete.
\end{proof}

Proposition \ref{proposition:rangeAccelerationTwoCars} establishes that the bearing-only strategy of \eqref{eq:mainResult} is a greedy strategy for maximising the instantaneous range acceleration $\ddot{r}(t)$ from hazards with finite turn rates.
Despite this property, the bearing-only strategy in general does not constitute a solution to the differential game \eqref{eq:problem} when it is formulated with a hazard that has a finite turn rate.
Indeed, strategies for maximising the minimum range from a hazard with a finite turn rate, require knowledge of the hazard's speed, turn rate, range, bearing, and heading (cf.~\cite{Olsder1978,Miloh1976} and \cite[Section 8.6]{Basar1999}).

Although the bearing-only strategy \eqref{eq:mainResult} is suboptimal for maximising the miss-distance (minimum range) when the hazard's turn rate is finite, it exhibits maximin robustness by maximising the miss-distance that can occur when the hazard is agile (i.e., capable of instantaneous heading changes).
Specifically, since an agile hazard capable of instantaneous heading changes can generate all other types of hazard motion, applying the bearing-only strategy \eqref{eq:mainResult} when the hazard has a finite turn rate will lead to a miss-distance that is greater than the miss-distance attained by applying it when the hazard can change heading instantaneously.
The difference between the miss-distances achieved by the bearing-only strategy \eqref{eq:mainResult} and the optimal strategies of \cite{Olsder1978,Miloh1976} for a hazard with a finite turn rate can be interpreted as the cost of the bearing-only strategy's robustness to not knowing the hazard's range, heading, and turning capabilities.

\section{Simulation Results}
\label{sec:results}

In this section, the effectiveness of the bearing-only strategy \eqref{eq:mainResult} is examined through software-in-the-loop simulations of test cases drawn from draft minimum operating performance standards for detect and avoid systems (i.e., Appendix P of \cite{RTCA2016}).

\subsection{Simulation Environment and Test Case Summary}
Our software-in-the-loop simulation environment is similar to those used in other works that have conducted software-in-the-loop simulation of guidance strategies (cf.~\cite{Garcia2010,Bittar2014,Garden2017,Agha2017}).
We employ a single Matlab/Simulink model interfaced to two instances of the commercial off-the-shelf flight simulator X-Plane 10 (cf.~\cite{XPlane}) --- one X-Plane instance to simulate the dynamics of the aircraft and the other to simulate the dynamics of the hazard.
Due to the lack of a standardised unmanned aircraft model in X-Plane, we use the Cessna 172SP aircraft model as a surrogate model for the unmanned aircraft and the Cirrus Vision SF50 aircraft model as the hazard.

The bearing-only strategy \eqref{eq:mainResult} and all of the autopilot control logic for the aircraft and hazard is implemented in Simulink so that the X-Plane instances only receive the desired throttle and control surface deflection commands, simulate the aircraft and hazard dynamics, and then return the aircraft and hazard state variables (i.e., airspeed, attitude, altitude, and position) to Simulink.
The aircraft and hazard autopilot architectures are implemented to match those of existing unmanned aircraft autopilots (see for example SLUGS \cite{Lizarraga2013} or MicroPilot \cite{MicroPilot}) with proportional-integral-derivative (PID) control loops that determine the aileron deflection angle from roll (roll-to-ailerons), the elevator deflection angle from pitch (pitch-to-elevator), the throttle command from airspeed (airspeed-to-throttle), and the rudder deflection angle from lateral acceleration (lateral-acceleration-to-rudder).
We shall focus on simulations with the hazard and aircraft in the same horizontal plane in order to illustrate the performance of the bearing-only strategy \eqref{eq:mainResult} without introducing the complexity of examining the best combination of vertical and horizontal controls to ensure that the aircraft and hazard remain in the plane defined by $\bm{V}_a$ and $\bm{R}$.
The output of the bearing-only strategy \eqref{eq:mainResult} is therefore converted into a commanded roll with the commanded values being $0^\circ$ and $\pm \, 60^\circ$ for $u_a(t) = 0$ and $u_a(t) = \pm 1$, respectively (corresponding to straight and level flight or steep level turns with a load factor of two).

To examine the effectiveness of the bearing-only strategy \eqref{eq:mainResult}, we selected eight test cases from Appendix P of \cite{RTCA2016} that capture a variety of head-on (H), converging (C), and overtaking (O) situations with the aircraft and hazard in the same horizontal plane.
We modified the test cases slightly to ensure feasibility with our aircraft and hazard models by keeping the velocity ratio $v_h/v_a$ approximately consistent but varying the speed of the aircraft.
The adapted cases are summarised in Table \ref{tbl:rtcaCases}.
The initial conditions in each test case are such that collision will occur if neither the aircraft nor hazard vary their velocities, with the intersect angles in Table \ref{tbl:rtcaCases} describing the angles between the hazard and aircraft paths at the point of collision.

For each test case in Table \ref{tbl:rtcaCases}, we consider initial ranges $r(0)$ of 1000m, 1500m, and 2000m since state-of-the-art aircraft detection systems based on electro-optical, infrared, and acoustic sensors are capable of reliable aircraft detections up to ranges of 2400m (cf.~\cite{Yu2015,James2018,James2018a}).
Consistent with our theoretical results, we omit consideration of sensor field-of-view limitations in our simulations (but we note that a 360$^\circ$ field of view is not unreasonable for acoustic sensors and is increasingly feasible with electro-optical and infrared sensors in arrays or equipped with specialised lenses \cite{Huang2017}).
Finally, since we have already established the optimality of the bearing-only strategy \eqref{eq:mainResult} for agile hazards that seek to minimise the miss-distance (and since representative numerical illustrations of the optimal trajectories are provided in Figs.~\ref{fig:fig3} -- \ref{fig:fig4}), in our simulations we shall consider a hazard that is non-responsive (i.e., it selects and maintains its heading independently of the aircraft).

\begin{table*}[!t]
\caption{Test cases adapted from Appendix P of \cite{RTCA2016}.}
\label{tbl:rtcaCases}
\begin{tabular}{@{}cccccl@{}}
\toprule
\textbf{Type} & \multicolumn{1}{c}{\textbf{ID}} & \multicolumn{1}{c}{\textbf{$v_a$ (knots)}} & \multicolumn{1}{c}{\textbf{$v_h/v_a$}} & \multicolumn{1}{c}{\textbf{Intersect Angle (degrees)}} & \multicolumn{1}{c}{\textbf{Initial Conditions Description}} \\ \midrule
\multirow{2}{*}{\textbf{Head-on}}    & H1                              & 50                                                 & 3                                               & 180                                                    & High speed encounter.                                       \\
                                     & H2                              & 42                                                  & 3.75                                               & 180                                                    & Low speed encounter.                                        \\ \cmidrule(l){1-6} 
\multirow{4}{*}{\textbf{Converging}} & C1                              & 50                                                 & 1.33                                               & 5                                                      & Hazard on right.                           \\
                                     & C6                              & 60                                                 & 1                                               & -60                                                    & Hazard on left.                           \\
                                     & C11                             & 60                                                 & 0.66                                              & 120                                                    & Hazard on right.\\         
                                     & C16                             & 60                                                 &   0.75                                             & -175                                                   & Hazard on left.
                                     \\ \cmidrule(l){1-6} 
\multirow{2}{*}{\textbf{Overtaking}} & O1                              & 42                                                  & 3.80                                               & 0                                                      & Hazard overtaking.                                          \\
                                     & O2                              & 60                                                 & 0.66                                               & 0                                                      & Aircraft overtaking.                                        
                                     \\\bottomrule
\end{tabular}
\end{table*}

\subsection{Head-on Cases}

Fig.~\ref{fig:HeadonCases} shows the simulated trajectories of the aircraft and hazard, as well as the ranges between them, in the two head-on test cases from an initial range of 2000m.
In both Cases H1 and H2, the hazard starts directly in front of the aircraft, and the aircraft's initial manoeuvre given by the bearing-only strategy \eqref{eq:mainResult} is to turn to the left (or right since both turns provide equal miss-distance in these cases).
After the hazard passes behind the aircraft, the aircraft (following the bearing-only strategy \eqref{eq:mainResult}) turns to keep the hazard directly behind it.
The aircraft is significantly slower than the hazard in both Cases H1 and H2 and so if the hazard employed its optimal strategy for minimising the range, it could ensure that $r(t) = 0$ for some $t \geq 0$.
However, due to the hazard maintaining its course, the bearing-only strategy \eqref{eq:mainResult} enables to the aircraft to manoeuvre such that $r(t) > 300$m in both cases.

In the simulations of Case H1 and Case H2 starting from initial ranges of 1000m and 1500m, the aircraft's response varies only in that it turns right to keep the hazard behind it sooner since the passage of the hazard behind the aircraft occurs earlier.
The minimum ranges between the aircraft and hazard in the simulations of Case H1 and Case H2 from different initial ranges are shown in Fig.~\ref{fig:HeadonCasesRanges}.
The bearing-only strategy ensures an approximately linear relationship between the initial and minimum ranges in these cases.

\begin{figure*}[tb!]
    \centering
    \begin{subfigure}[b]{3.14in}
        \includegraphics[width=\textwidth]{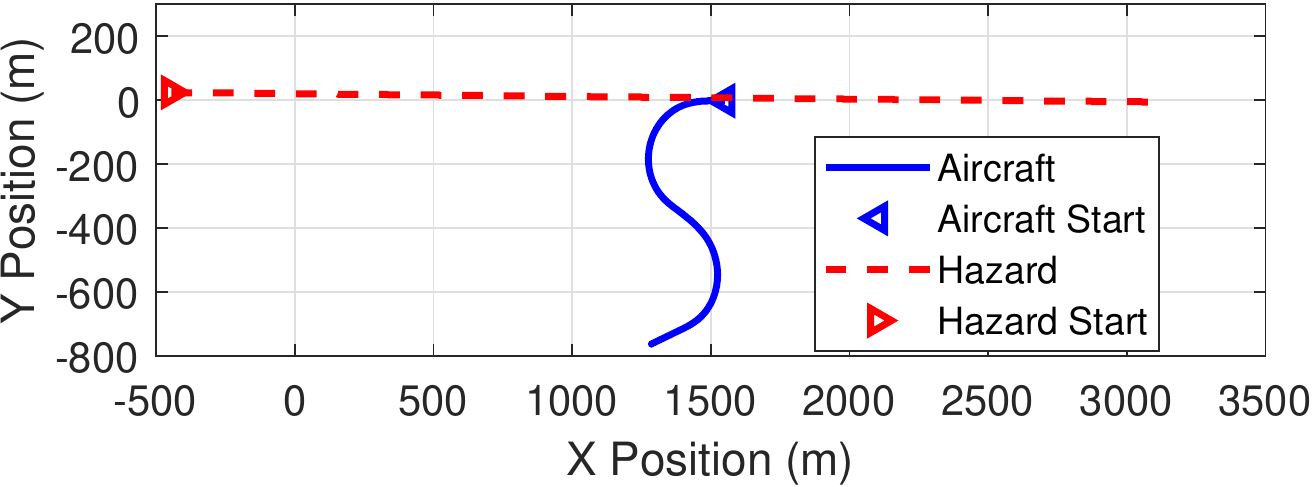}
        \caption{}
    \end{subfigure}
    \begin{subfigure}[b]{3.14in}
        \includegraphics[width=\textwidth]{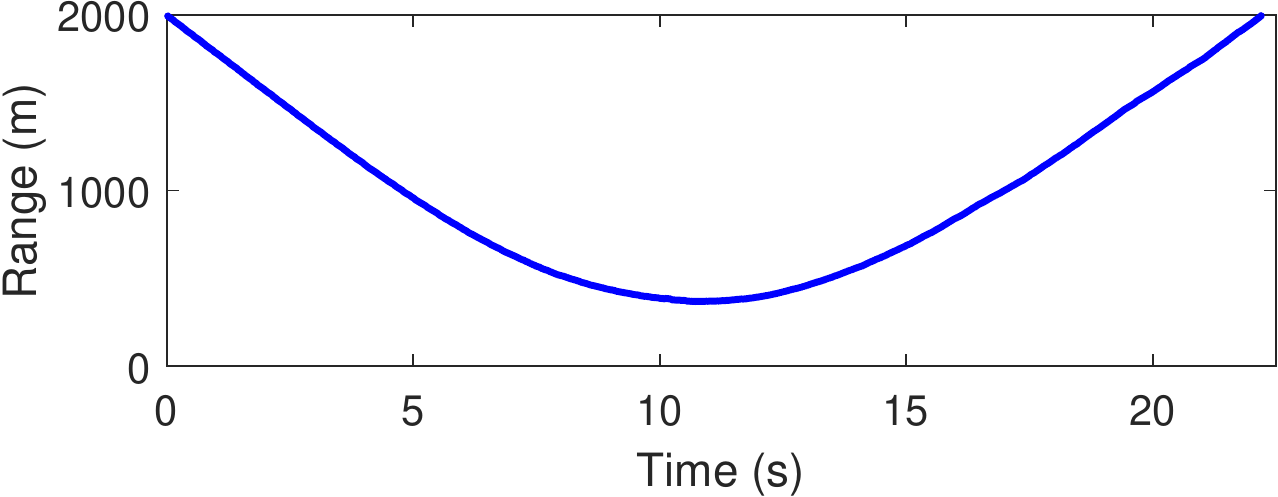}
        \caption{}
    \end{subfigure}\\
    \begin{subfigure}[b]{3.14in}
        \includegraphics[width=\textwidth]{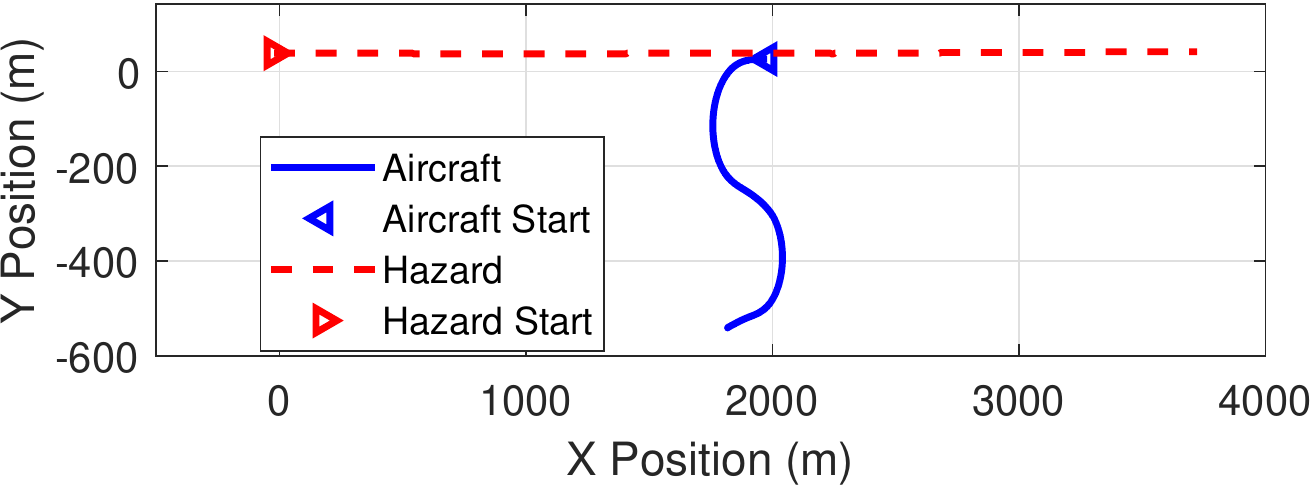}
        \caption{}
    \end{subfigure}
    \begin{subfigure}[b]{3.14in}
        \includegraphics[width=\textwidth]{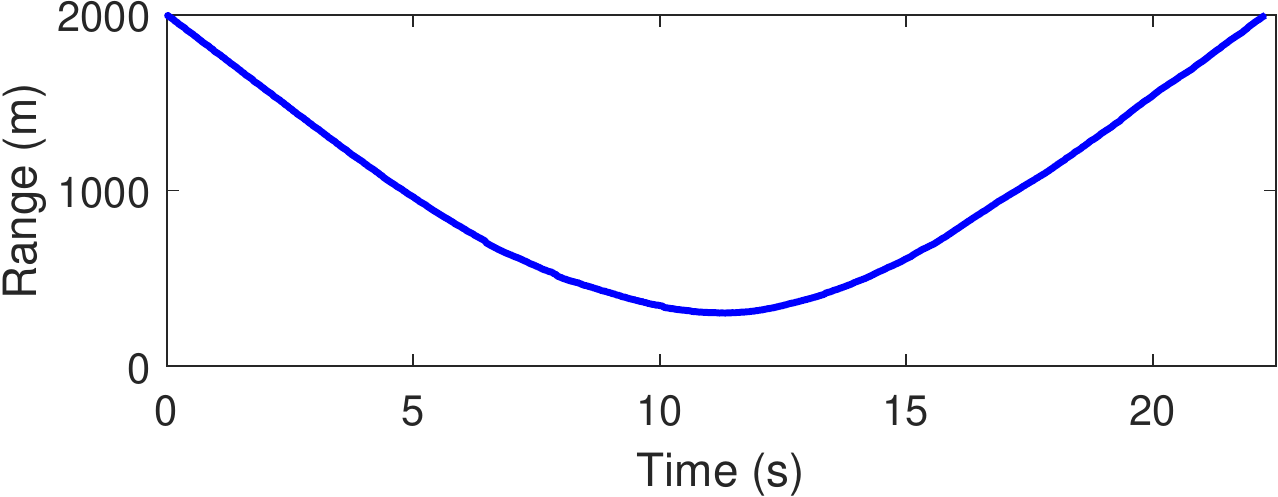}
        \caption{}
    \end{subfigure}
    \caption{Head-on cases from initial range 2000m: (a) Trajectories and (b) Range for H1, and (c) Trajectories and (d) Range for H2.}
    \label{fig:HeadonCases}
\end{figure*}

\begin{figure*}[tb!]
    \centering
    \begin{subfigure}[b]{3.14in}
        \includegraphics[width = \textwidth]{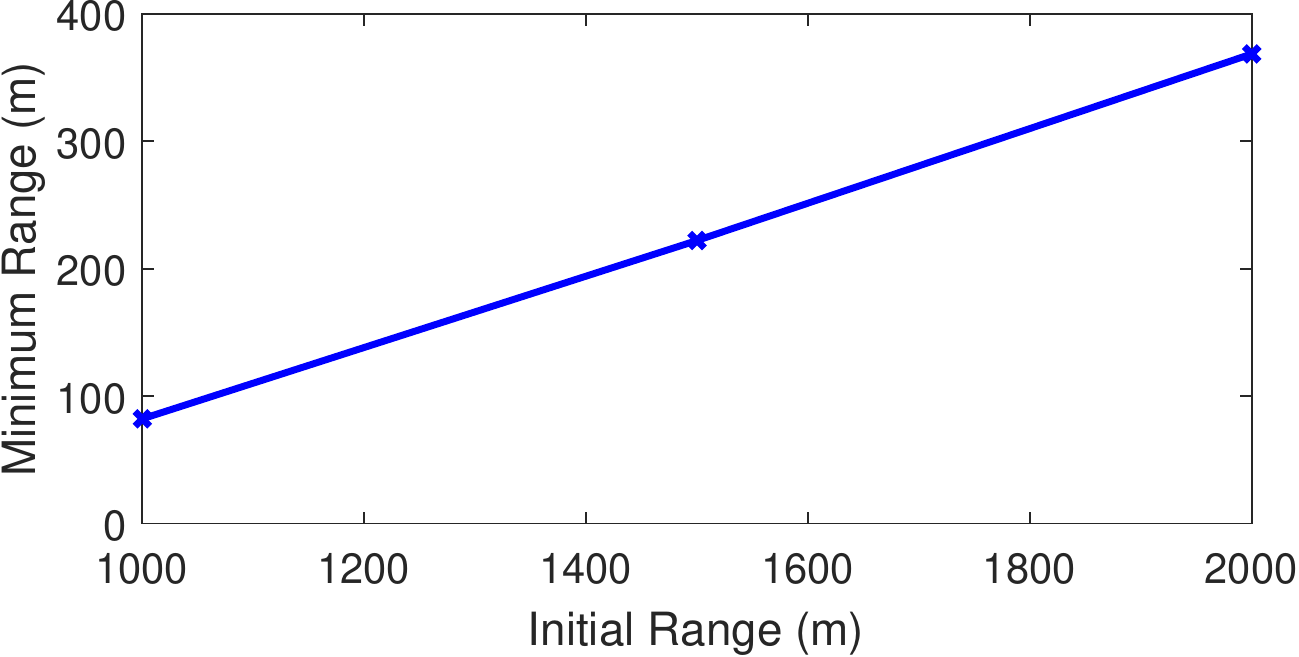}
        \caption{}
    \end{subfigure}
    \begin{subfigure}[b]{3.14in}
        \includegraphics[width=\textwidth]{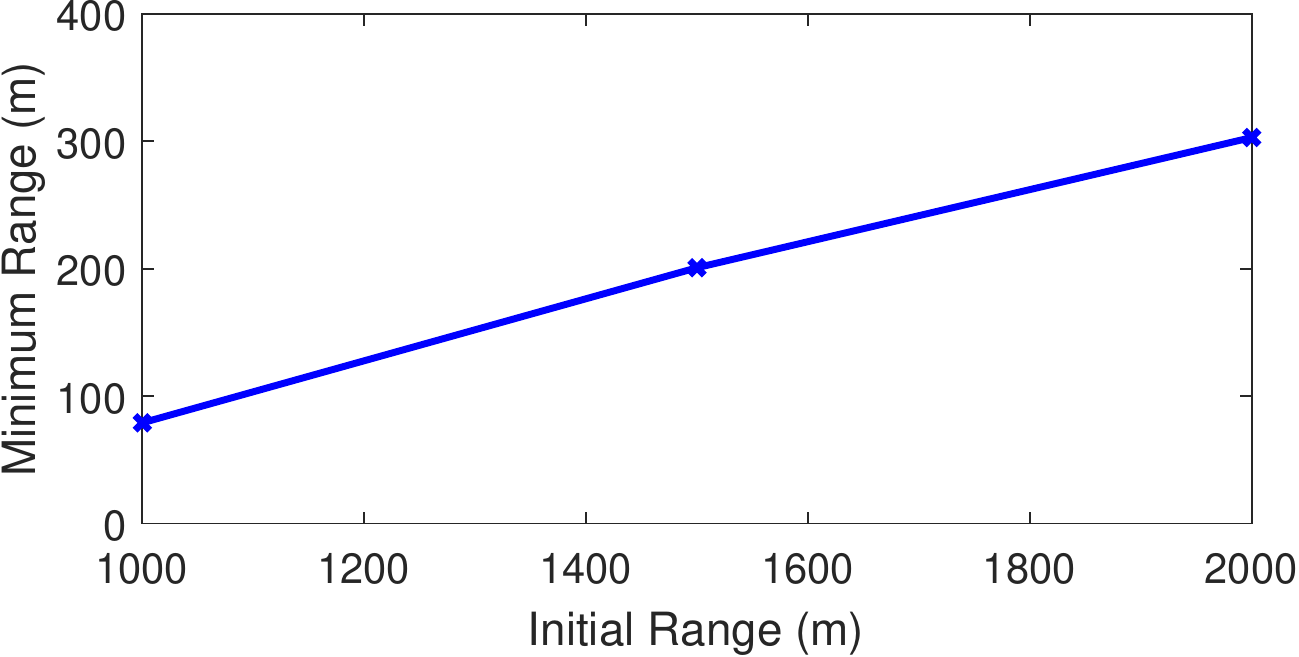}
        \caption{}
    \end{subfigure}\\
    \caption{Head-on Cases: Minimum Range $r(T)$ versus Initial Range $r(0)$ for (a) Case H1 and (b) Case H2.}
    \label{fig:HeadonCasesRanges}
\end{figure*}

\subsection{Converging Cases}

The simulated trajectories of the aircraft and hazard in the converging cases with initial ranges of 2000m are shown in Fig.~\ref{fig:ConvergingCases}.
From Fig.~\ref{fig:ConvergingCases}, we see that the bearing-only strategy \eqref{eq:mainResult} enables the aircraft to ensure that $r(t) > 500$m in all four cases (including those with faster hazards).
In all cases, an initial turn in the opposite direction to that prescribed by the bearing-only strategy \eqref{eq:mainResult} would lead to a smaller miss-distance provided the hazard does not also change its heading.
The bearing-only strategy \eqref{eq:mainResult} in these cases therefore manages to maximise the miss-distance, despite in general being suboptimal when the hazard's turn rate is finite (as we discussed after Proposition \ref{proposition:rangeAccelerationTwoCars}).
It is interesting to note that the more gradual second turns in Case C1 and Case C6 occur due to the hazard crossing behind the aircraft and the bearing-only strategy attempting to keep the hazard behind the aircraft by switching the commanded roll between $0^\circ$ and $\pm 60^\circ$.

As in the head-on cases, the trajectories of the aircraft and hazard in simulations of the converging cases from initial ranges of 1000m and 1500m have largely the same features as when the initial range is 2000m.
The minimum ranges from different initial ranges for these converging cases are shown in Fig.~\ref{fig:ConvergingCasesRanges}.
The bearing-only strategy again ensures an approximately linear relationship between the initial and minimum ranges in these cases.

%Nevertheless, from a purely collision avoidance perspective, the bearing-only strategy \eqref{eq:mainResult} provides satisfactory performance in all simulated cases, including those with faster hazards.

\begin{figure*}[tb!]
    \centering
    \begin{subfigure}[b]{3.14in}
        \includegraphics[width=\textwidth]{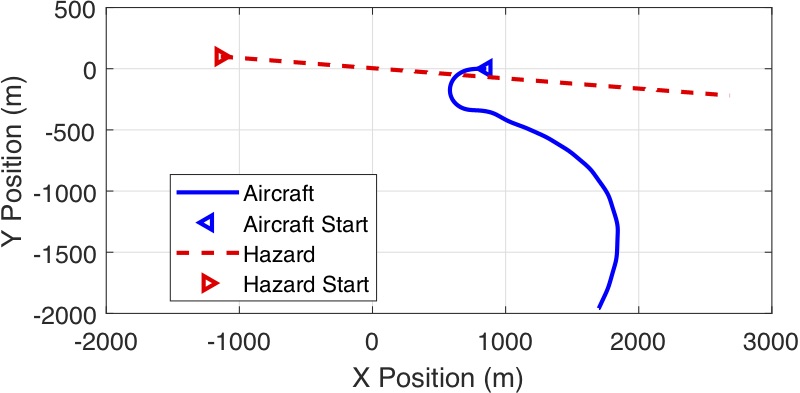}
        \caption{}
    \end{subfigure}
    \begin{subfigure}[b]{3.14in}
        \includegraphics[width=\textwidth]{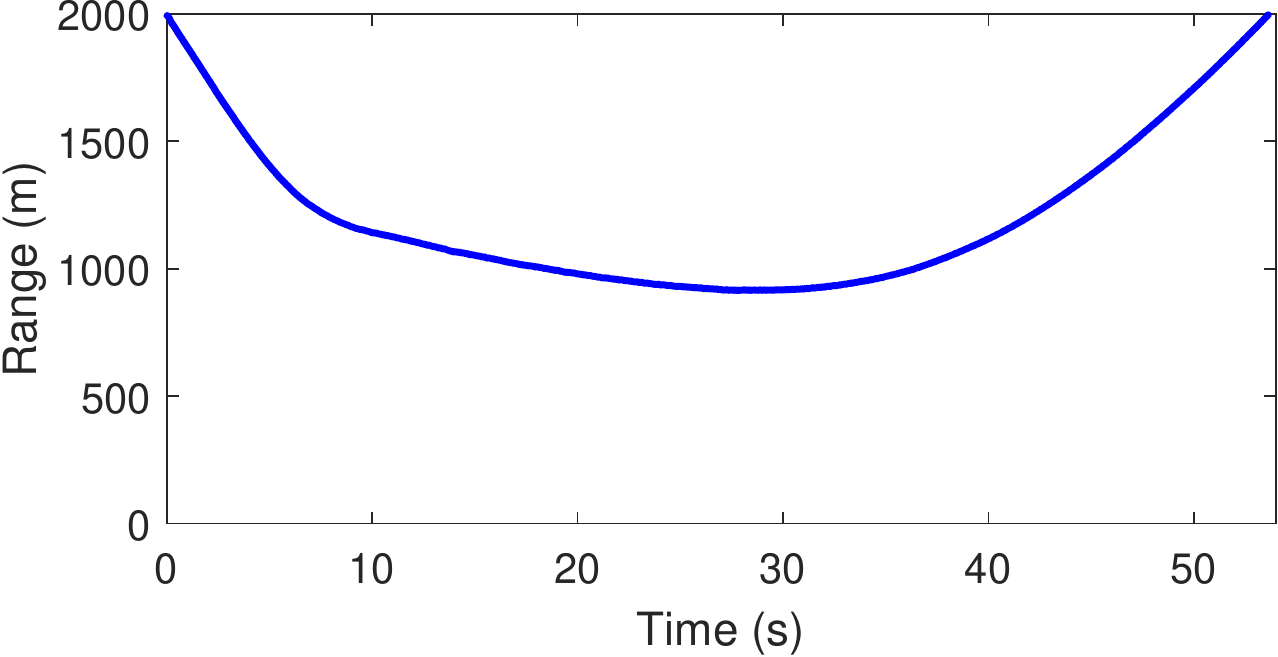}
        \caption{}
    \end{subfigure}\\
    \begin{subfigure}[b]{3.14in}
        \includegraphics[width=\textwidth]{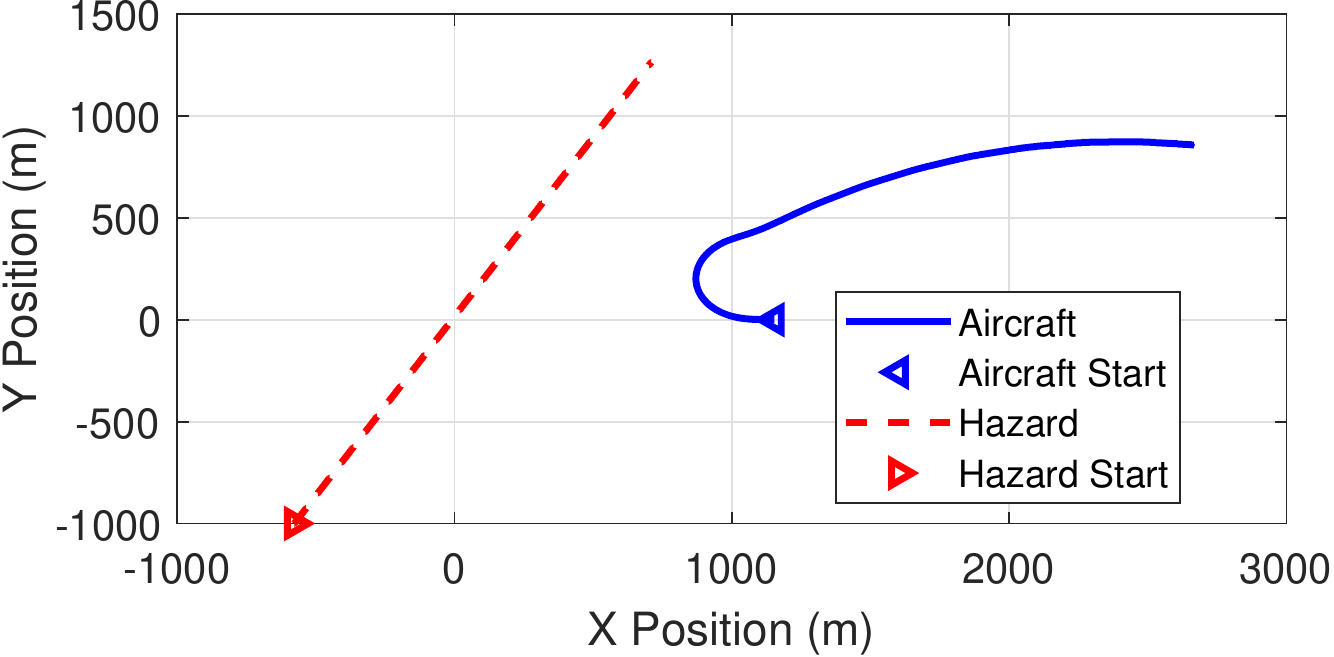}
        \caption{}
    \end{subfigure}
    \begin{subfigure}[b]{3.14in}
        \includegraphics[width=\textwidth]{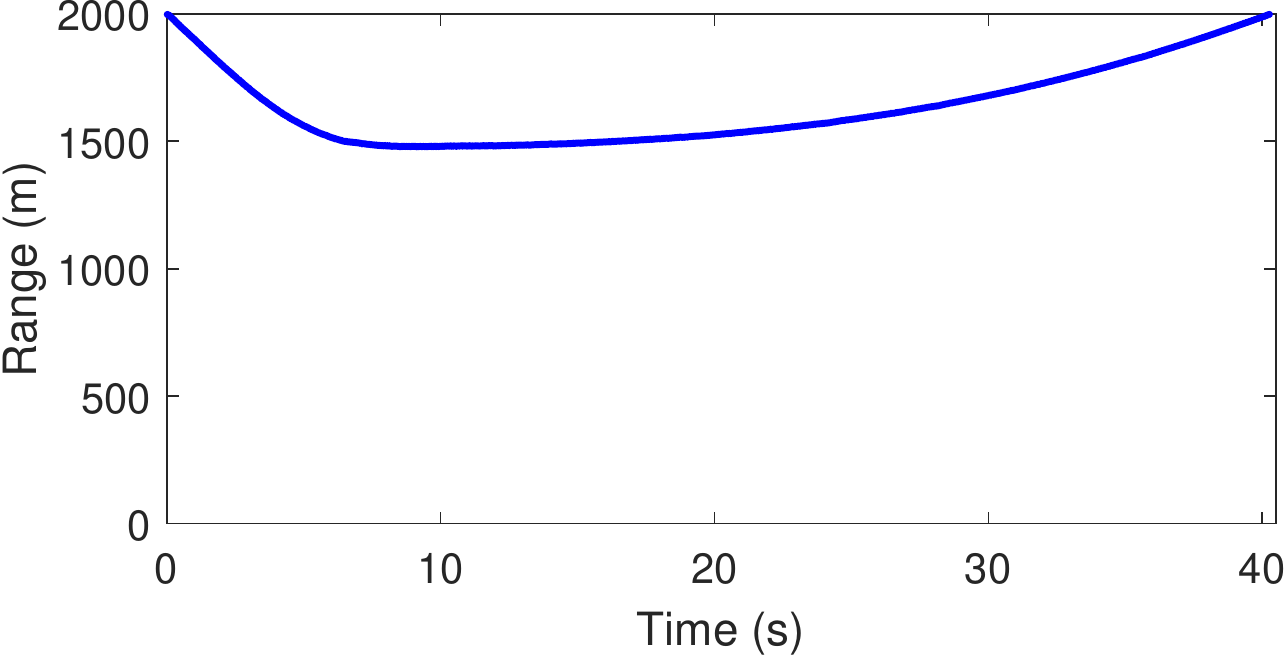}
        \caption{}
    \end{subfigure}\\
    \begin{subfigure}[b]{3.14in}
        \includegraphics[width=\textwidth]{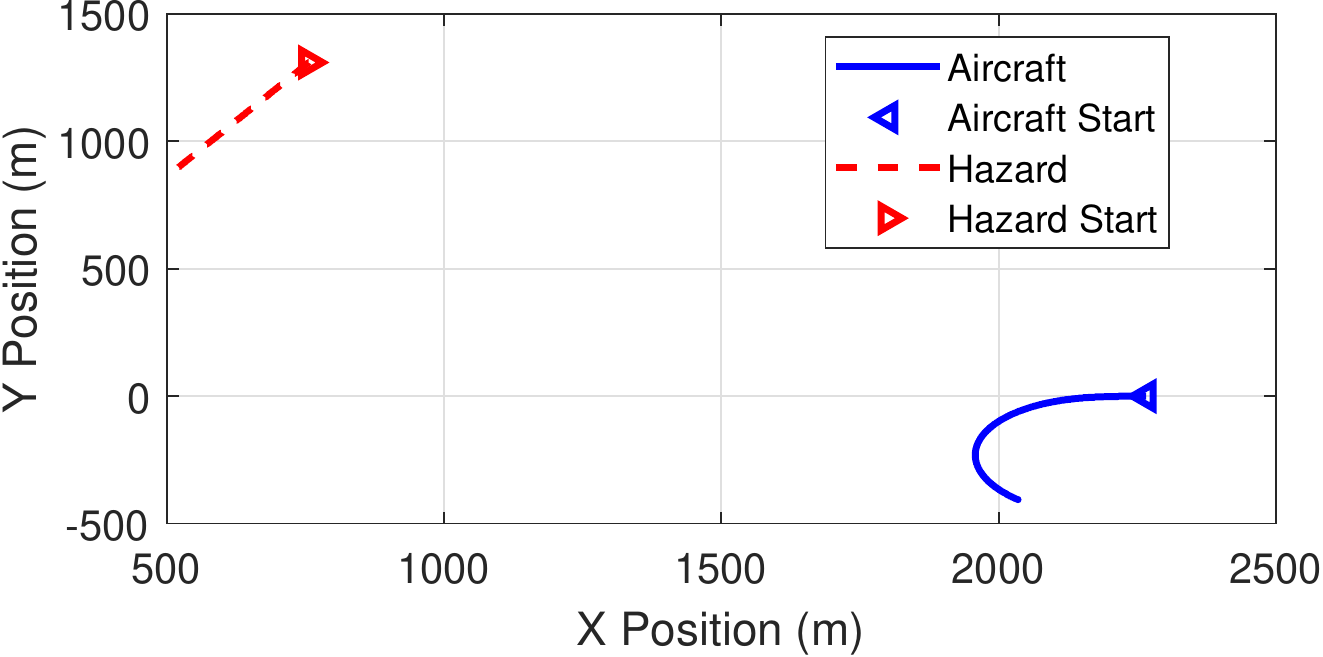}
        \caption{}
    \end{subfigure}
    \begin{subfigure}[b]{3.14in}
        \includegraphics[width=\textwidth]{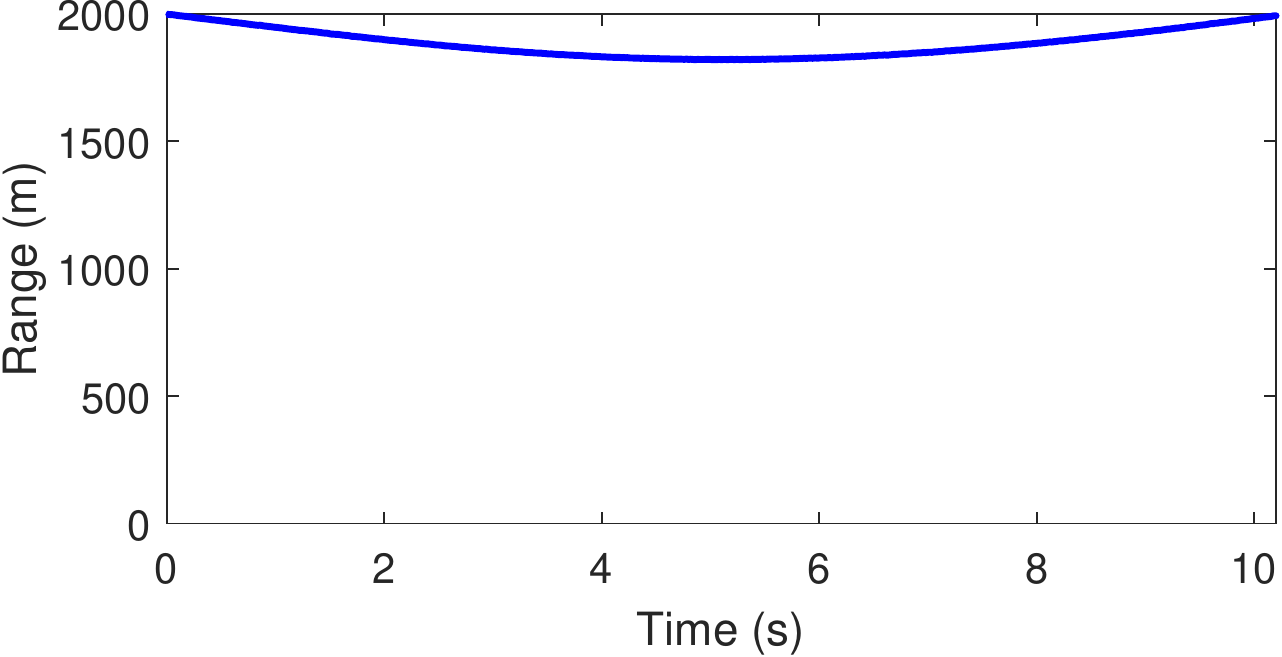}
        \caption{}
    \end{subfigure}\\
        \begin{subfigure}[b]{3.14in}
        \includegraphics[width=\textwidth]{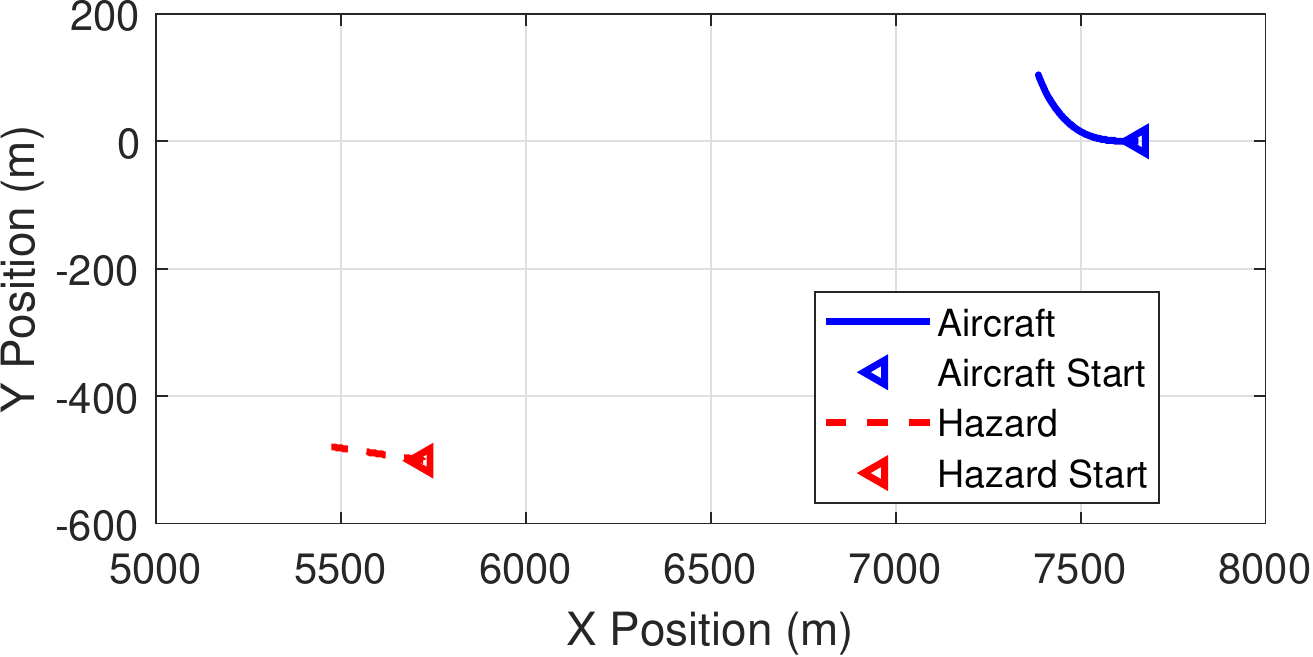}
        \caption{}
    \end{subfigure}
    \begin{subfigure}[b]{3.14in}
        \includegraphics[width=\textwidth]{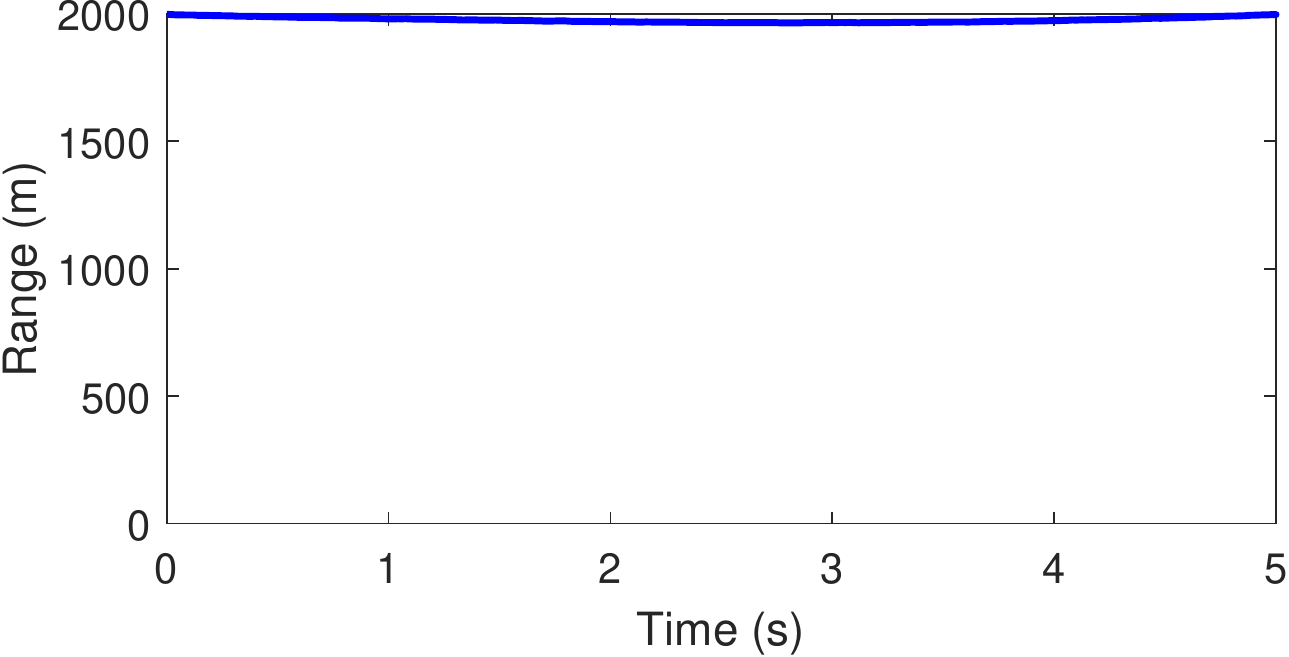}
        \caption{}
    \end{subfigure}
    \caption{Converging Cases from initial range 2000m: (a) Trajectories and (b) Range for C1, (c) Trajectories and (d) Range for C6, (e) Trajectories and (f) Range for C11, and (g) Trajectories and (h) Range for C16.}
    \label{fig:ConvergingCases}
\end{figure*}

\begin{figure*}[tb!]
    \centering
    \begin{subfigure}[b]{3.14in}
        \includegraphics[width = \textwidth]{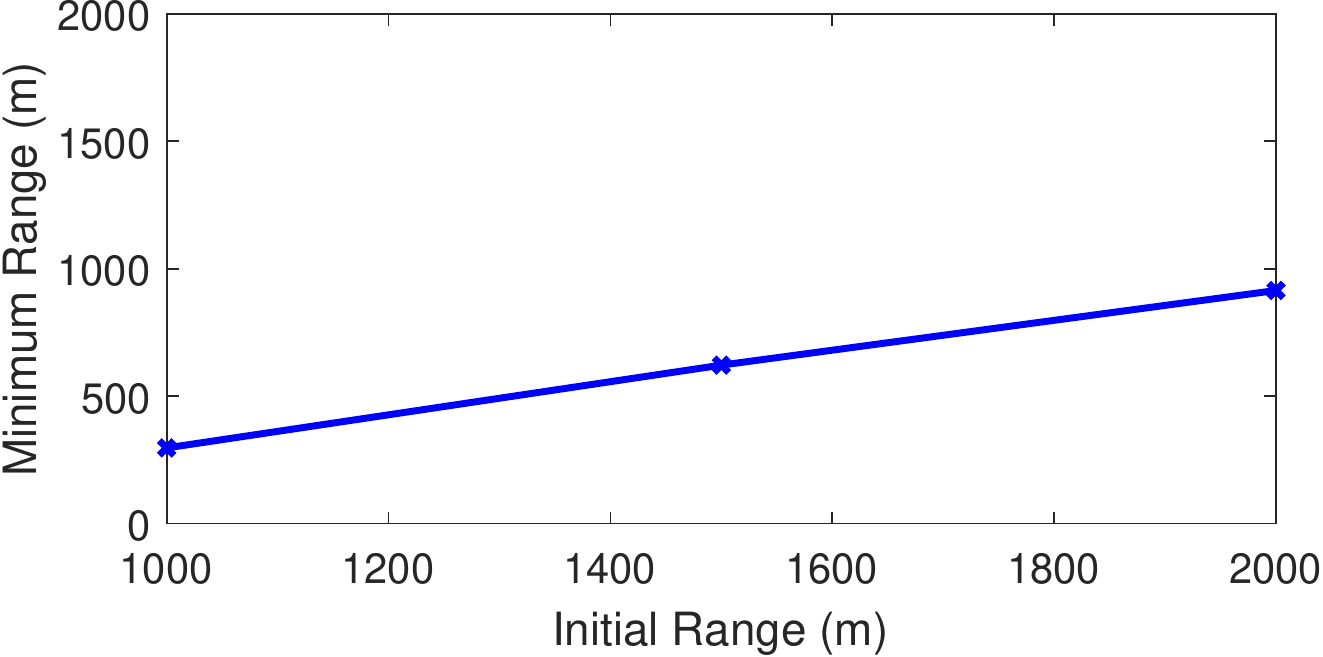}
        \caption{}
    \end{subfigure}
    \begin{subfigure}[b]{3.14in}
        \includegraphics[width=\textwidth]{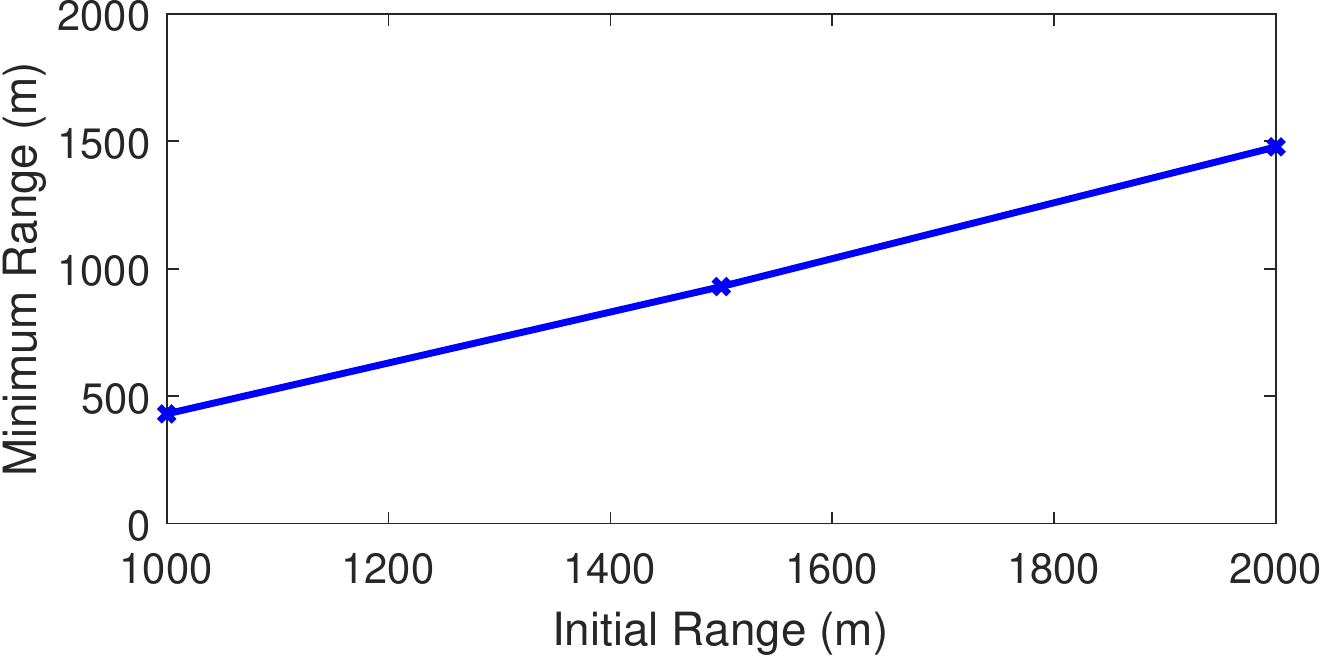}
        \caption{}
    \end{subfigure}\\
        \begin{subfigure}[b]{3.14in}
        \includegraphics[width = \textwidth]{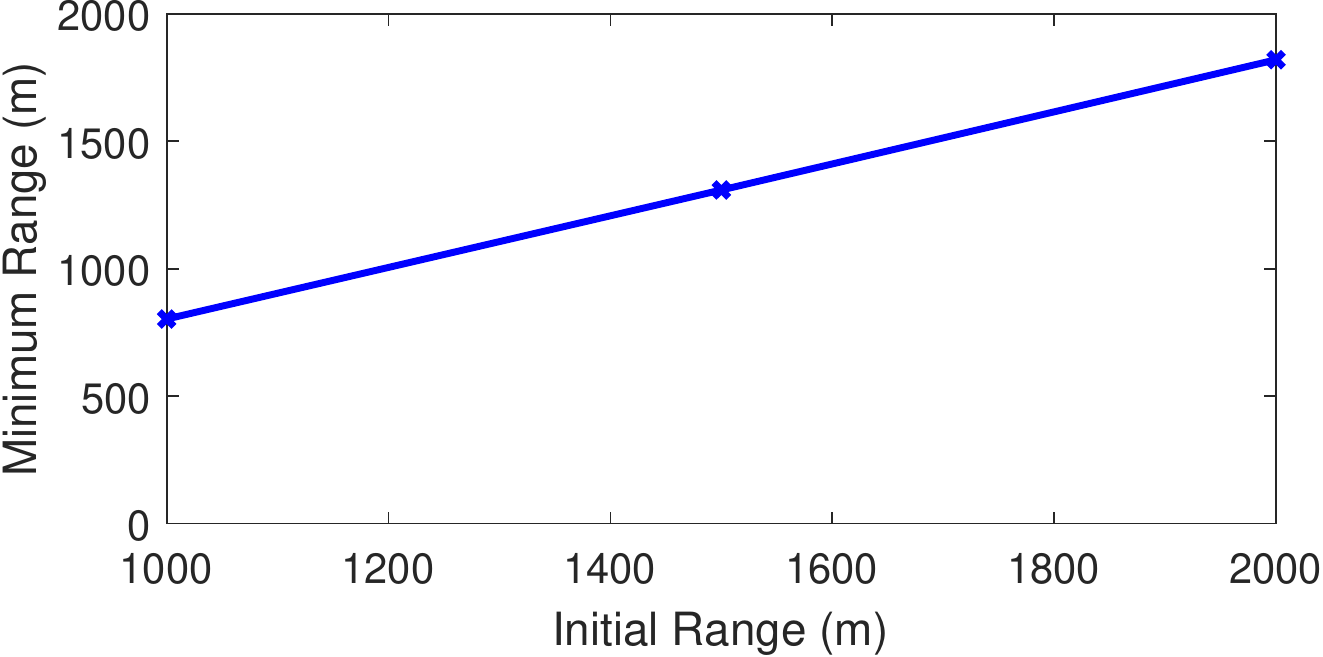}
        \caption{}
    \end{subfigure}
    \begin{subfigure}[b]{3.14in}
        \includegraphics[width=\textwidth]{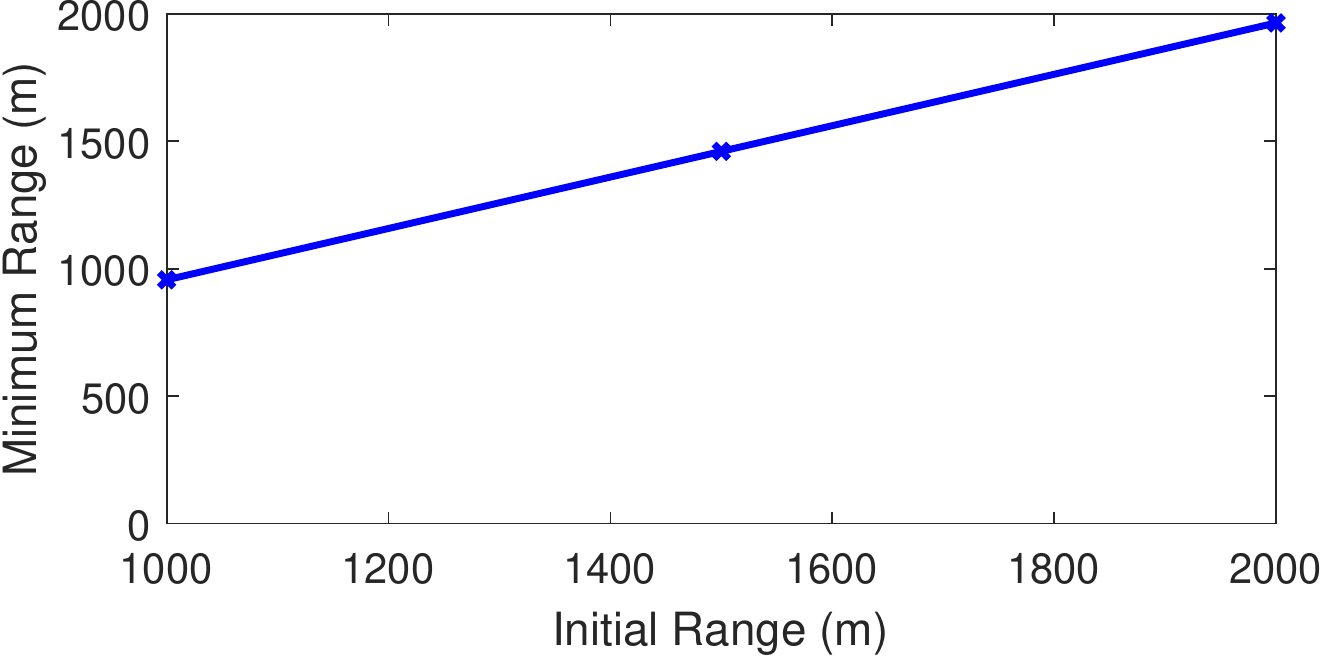}
        \caption{}
    \end{subfigure}
    \caption{Converging Cases: Minimum Range $r(T)$ versus Initial Range $r(0)$ for (a) Case C1, (b) Case C6, (c) Case C11, and (b) Case C16.}
    \label{fig:ConvergingCasesRanges}
\end{figure*}

\subsection{Overtaking Cases}

The simulated trajectories and ranges for the overtaking cases from an initial range of 2000m are presented in Fig.\ref{fig:OvertakingCases}.
In Case O2, the aircraft overtakes the slower hazard (and avoids colliding with it) by employing a left turn as prescribed by the bearing-only strategy \eqref{eq:mainResult}.
However, in Case O1, the aircraft and hazard collide since the aircraft maintains its course in accordance with the bearing-only strategy \eqref{eq:mainResult} and the hazard is faster than the aircraft.
The bearing-only strategy \eqref{eq:mainResult} has however offered protection against the hazard in Case O1 employing its optimal strategy for minimising the miss-distance since if the aircraft were to turn, and the hazard were to employ its optimal strategy, the range would decrease at a faster rate and collision would occur earlier.
By following the bearing-only strategy \eqref{eq:mainResult} in Case O1, the aircraft has therefore succeeded in ensuring that the range decreases at the slowest possible rate. 
We also note that by maintaining its course whilst being overtaken, the aircraft has obeyed the rules of the air (cf.~\cite{ICAO2005}).

The trajectories of the aircraft and hazard in simulations of the overtaking cases from initial ranges of 1000m and 1500m are essentially identical to those where the initial range is 2000m.
The minimum ranges from different initial ranges for these overtaking cases are shown in Fig.~\ref{fig:OvertakingCasesRanges}.
When the aircraft is being overtaken in Case O1, no initial range results in a nonzero minimum range.
However, when the aircraft is overtaking the hazard in Case O2, the bearing-only strategy ensures an approximately linear relationship between the initial and minimum ranges.

\begin{figure*}[tb!]
    \centering
    \begin{subfigure}[b]{3.14in}
        \includegraphics[width=\textwidth]{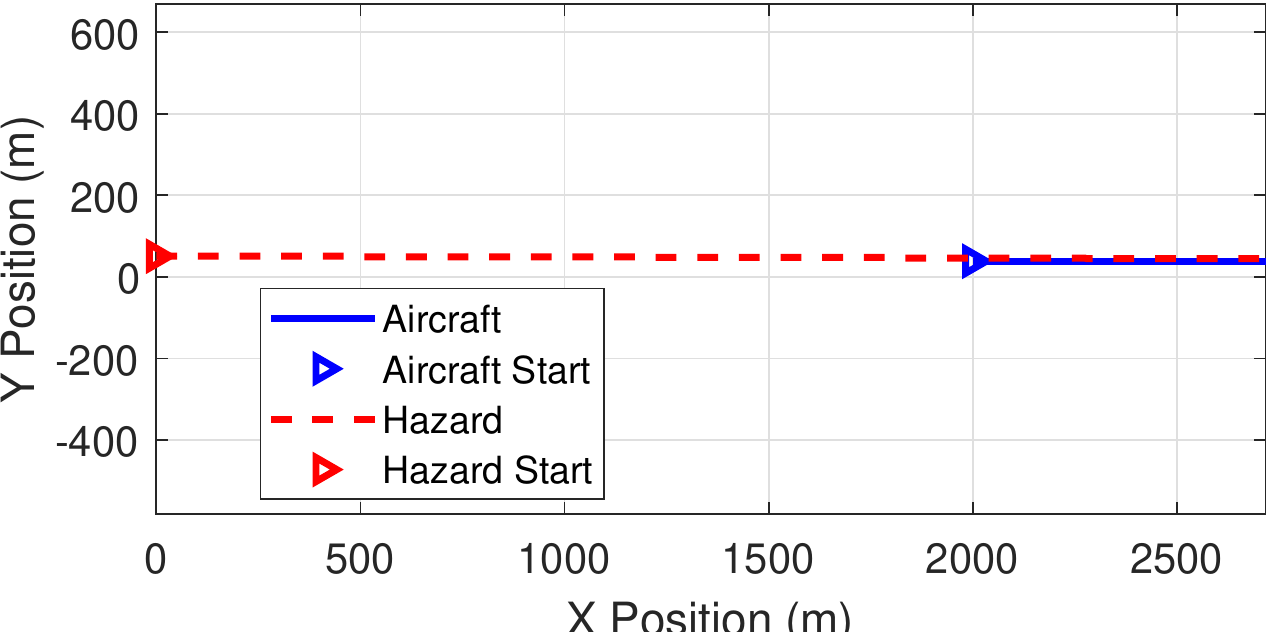}
        \caption{}
    \end{subfigure}
    \begin{subfigure}[b]{3.14in}
        \includegraphics[width=\textwidth]{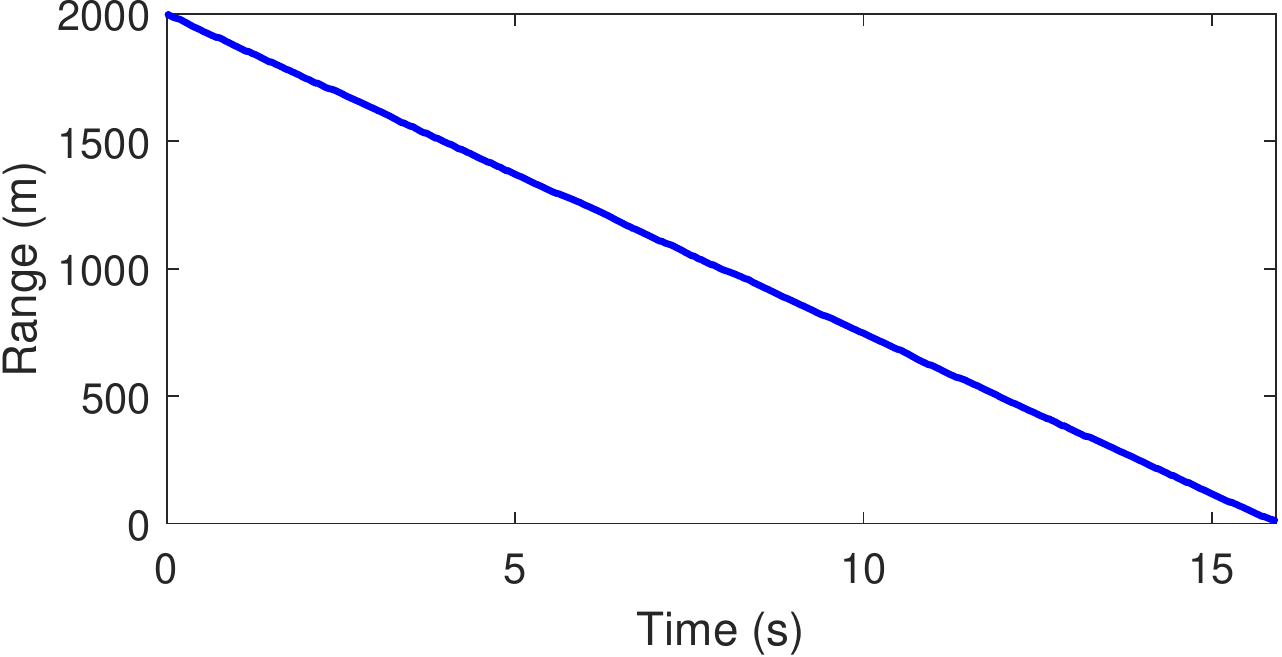}
        \caption{}
    \end{subfigure}\\
    \begin{subfigure}[b]{3.14in}
        \includegraphics[width=\textwidth]{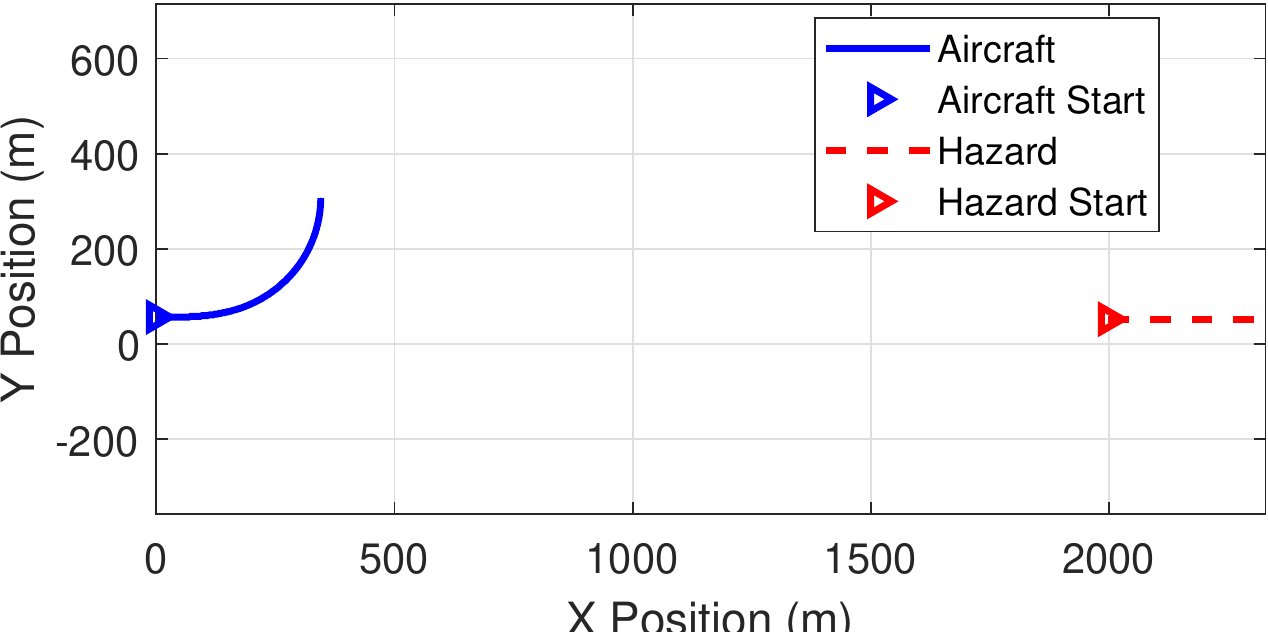}
        \caption{}
    \end{subfigure}
    \begin{subfigure}[b]{3.14in}
        \includegraphics[width=\textwidth]{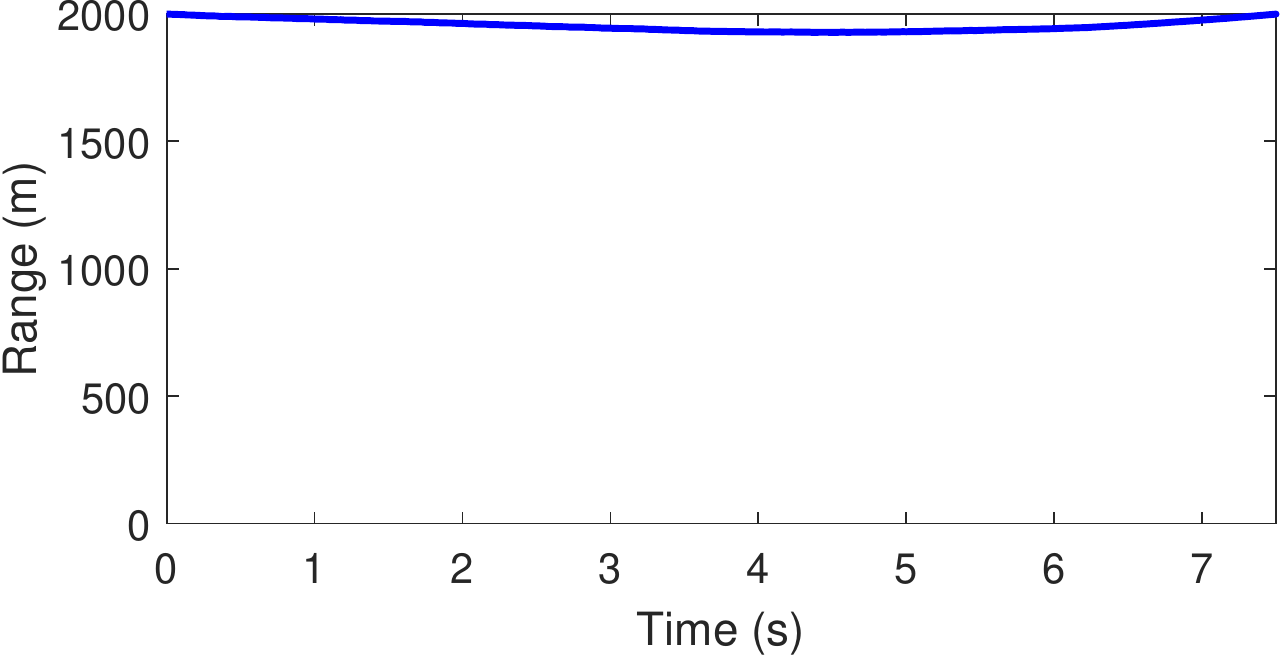}
        \caption{}
    \end{subfigure}
    \caption{Overtaking Cases from initial range 2000m: (a) Trajectories and (b) Range for O1, and (c) Trajectories and (d) Range for O2.}
    \label{fig:OvertakingCases}
\end{figure*}

\begin{figure*}[tb!]
    \centering
    \begin{subfigure}[b]{3.14in}
        \includegraphics[width = \textwidth]{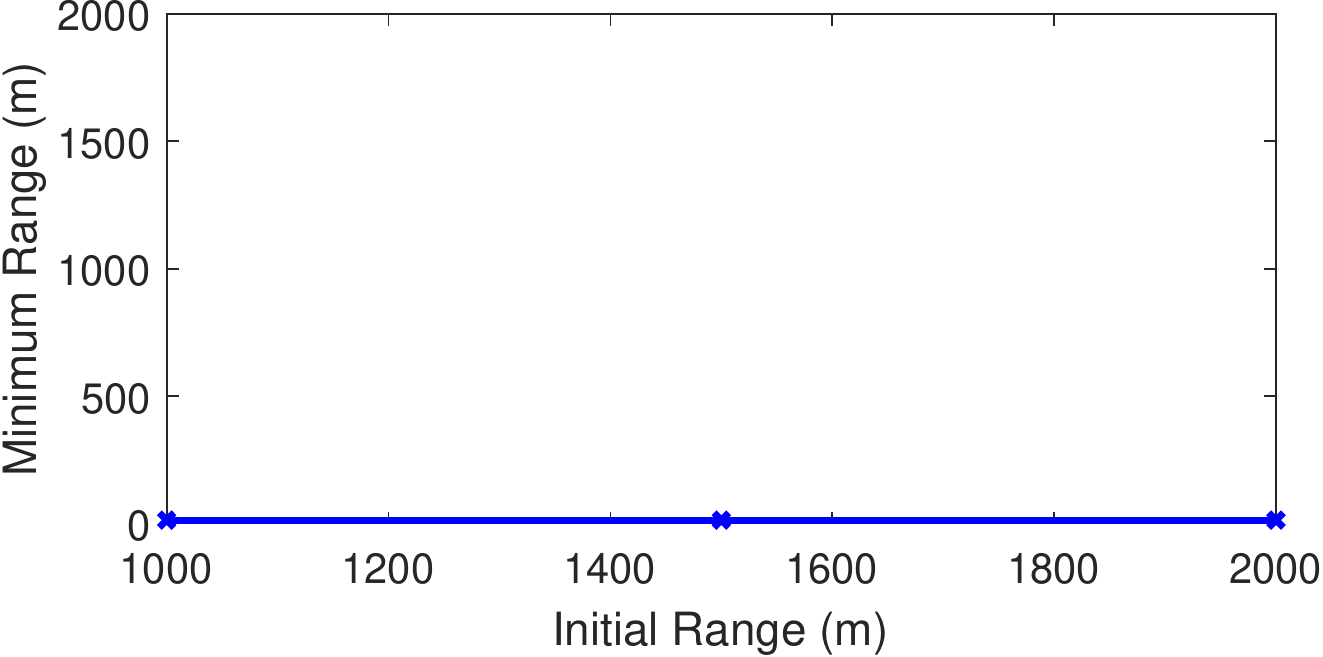}
        \caption{}
    \end{subfigure}
    \begin{subfigure}[b]{3.14in}
        \includegraphics[width=\textwidth]{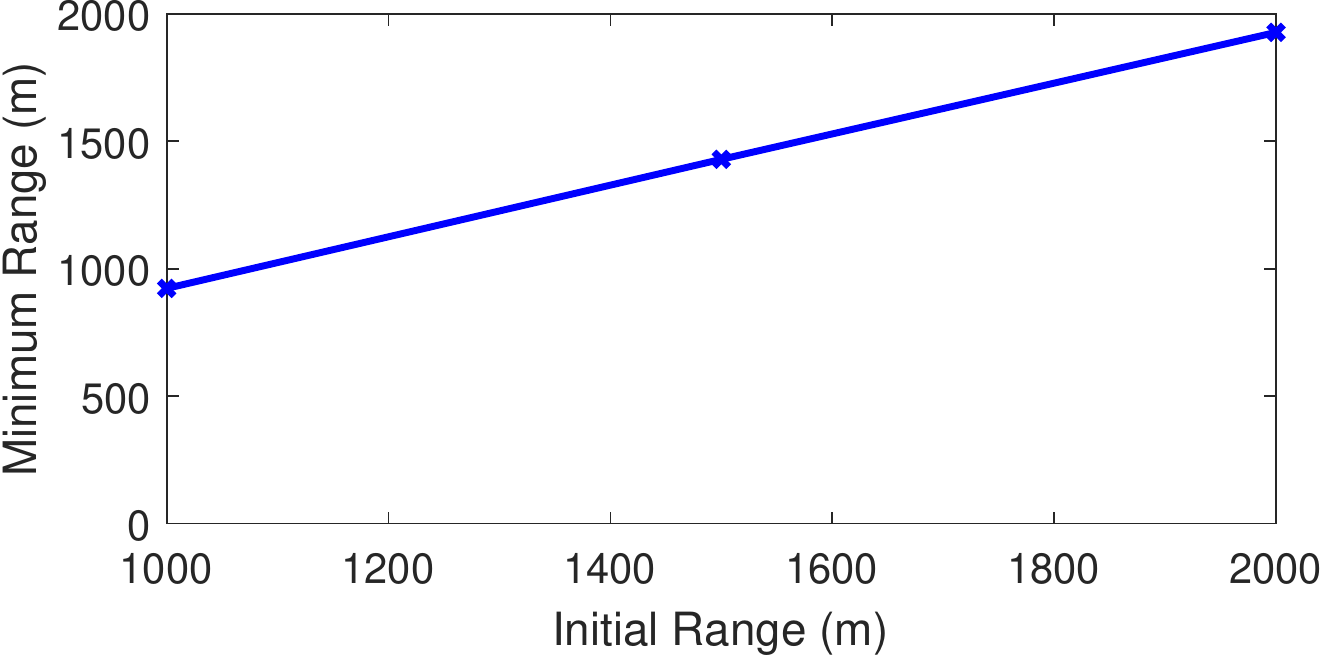}
        \caption{}
    \end{subfigure}\\
    \caption{Overtaking Cases: Minimum Range $r(T)$ versus Initial Range $r(0)$ for (a) Case O1 and (b) Case O2.}
    \label{fig:OvertakingCasesRanges}
\end{figure*}

\subsection{Simulation Discussion and Summary}
A near mid-air collision is typically defined as when the aircraft and hazard come within 500 feet (152.4m) horizontally of each other \cite{Billingsley2012}.
With the exception of the Case O1, the bearing-only strategy \eqref{eq:mainResult} enables the avoidance of near mid-air collisions in the simulated cases of Table \ref{tbl:rtcaCases} from an initial range of 2000m.
In Case O1, the aircraft collides with the hazard because the bearing-only strategy \eqref{eq:mainResult} is concerned that by manoeuvring the aircraft, the hazard can change its course to collide sooner (if the hazard were interested in minimising the miss-distance).
We also note that by maintaining its course whilst being overtaken, the aircraft has obeyed the right of way provisions in the rules of the air (cf.~\cite{ICAO2005}).
The simulations therefore suggest that the bearing-only strategy \eqref{eq:mainResult} is capable of enabling collision avoidance at ranges where existing state-of-the-art detection systems based on electro-optical, infrared, and acoustic sensors are capable of providing aircraft detections (and hence bearing measurements) (cf.~\cite{Yu2015,James2018,James2018a}).

\section{Conclusion}
\label{sec:conclusion}

This paper identifies a novel bearing-only collision avoidance strategy that can be conceived as a solution of a game in which an unmanned aircraft seeks to maximise the miss-distance to an agile hazard capable of instantaneously changing its heading to reduce the miss-distance.
The strategy is specifically shown to arise as the solution to a zero-sum differential game in three dimensions, and to maximise the instantaneous acceleration in range between the aircraft and hazards with finite turn rates.
The significance of the bearing-only strategy lies in its applicability to selecting collision avoidance manoeuvres for unmanned aircraft detect and avoid systems that (either by fault or design) may only have access to a single sensor capable of direct measurement of bearing (e.g. an electro-optical, infrared, or acoustic sensor).
Simulations show that the strategy is capable of avoiding collisions with hazards that may also maintain their course, rather than manoeuvring to minimise the miss-distance.

A number of extensions of the bearing-only strategy presented in this paper and its associated analysis under our differential-game formulation of collision avoidance remain to be explored including: investigation of the impact of discritisation and/or noise on the available bearing angles; the consideration of sensor field of view and aircraft performance constraints (e.g., the prioritisation of horizontal or vertical manoeuvres); and, the consideration of multiple hazards.
Importantly, many of these possible extensions may be viewed as introducing extra constraints into our differential-game formulation of collision avoidance.
This paper has therefore established a theoretical bound on the (unconstrained) performance achievable by bearing-only collision avoidance approaches against which subsequent (constrained) extensions can be evaluated.

% \section*{Appendix}

% An Appendix, if needed, appears \textbf{before} research funding information and other acknowledgments.

\section*{Funding Sources}
This work was supported by the Queensland Government Department of Science, Information Technology and Innovation (DSITI) and Boeing Research And Technology - Australia through an Advance Queensland Research Fellowship (AQRF07616-12RD2) to T.~L.~Molloy.

\section*{Acknowledgments}
We gratefully acknowledge Grace Garden, Dr Neale Fulton, and Prof Jason Ford for fruitful discussions on collision avoidance that helped to set the context of the paper.

\bibliography{sample}

\end{document}